\documentclass[11pt]{article}
\usepackage{amssymb,amsmath,latexsym,amscd,amsfonts, bbm}  
\usepackage{graphicx}  
\usepackage[usenames]{color}
\usepackage{umoline}

\newtheorem{theorem}{Theorem}[section]  
\newtheorem{remark}{Remark}[theorem]  
  
\newtheorem{definition}[theorem]{Definition}  
\newtheorem{deff}[theorem]{Definition}

\newtheorem{lemma}[theorem]{Lemma}

\newtheorem{corollary}[theorem]{Corollary}  
\newtheorem{corol}[theorem]{Corollary}

\newcommand{\qedsymb}{\hfill{\rule{2mm}{2mm}}}  
\newenvironment{proof}[1][]{\begin{trivlist}  
\item[\hspace{\labelsep}{\bf\noindent Proof#1:\/}] 
}{\qedsymb\end{trivlist}}

\newcommand{\znote}[1]{}
\newcommand{\inote}[1]{}
\newcommand{\ignore}[1]{}

\newcommand{\BQP}{\textsf{BQP}}
\newcommand{\BPP}{\textsf{BPP}}
\newcommand{\QMA}{\textsf{QMA}}
\newcommand{\QCMA}{\textsf{QCMA}}
\newcommand{\NP}{\textsf{NP}}
\newcommand{\SharpP}{\textsf{\#P}}

\newcommand{\B}[1]{\mathbf{#1}}

\newcommand{\norm}[1]{{\| #1 \|}}  
  
\newcommand{\ket}[1]{{ |{#1} \rangle }}  
\newcommand{\bra}[1]{{ \langle {#1} | }}

\newcommand{\poly}{\mathrm{poly}} 
\newcommand{\EqDef}{\stackrel{\mathrm{def}}{=}}

\newcommand{\Eq}[1]{Eq.~(\ref{#1})}
\newcommand{\Fig}[1]{Fig.~\ref{#1}}
\newcommand{\Def}[1]{Def.~\ref{#1}}

\newcommand{\Sec}[1]{Sec.~\ref{#1}}
\newcommand{\Ref}[1]{\cite{#1}}
\newcommand{\Col}[1]{Corollary~\ref{#1}}
\newcommand{\Thm}[1]{Theorem~\ref{#1}}
\newcommand{\App}[1]{Appendix~\ref{#1}}

\newcommand{\Id}{\mathbbm{1}}
\newcommand{\BBC}{\mathbbm{C}}
\newcommand{\BBR}{\mathbbm{R}}
\newcommand{\BBB}{\mathbbm{B}}
\newcommand{\CM}{\mathcal{M}}
\newcommand{\CO}{\mathcal{O}}

\renewcommand{\epsilon}{\varepsilon}

\begin{document}

\title{Quantum computation and the evaluation of tensor networks.} 
\author{  
 Itai Arad
 \thanks{Department of Electrical Engineering and Computer Sciences, 
   University of California at Berkeley, CA.  \newline E-mail:
   \texttt{arad.itai@gmail.com}}  
 \\and  
 Zeph Landau \thanks{ Department of Electrical Engineering and Computer Sciences,
 University of California at Berkeley,CA. } } \maketitle 
 
\noindent 
\begin{abstract} 
  
  We present a quantum algorithm that additively approximates the
  value of a tensor network to a certain scale. When combined with
  existing results, this provides a complete problem for quantum
  computation.  The result is a simple new way of looking at quantum
  computation in which unitary gates are replaced by tensors and
  time is replaced by the order in which the tensor-network is
  ``swallowed''.  We use this result to derive new quantum
  algorithms that approximate the partition function of a variety of
  classical statistical mechanics models, including the Potts model.  
  
\end{abstract}

\section{Introduction}

The discovery by Peter Shor in 1994 of a quantum algorithm for
factoring $n$ digit numbers in $poly(n)$ steps stimulated a large
amount of interest in the power of quantum computation
\cite{ref:Sho97}. Since then, the search for quantum algorithms that
provide exponential speedup over the best known classical algorithms
has yielded a number of results: algorithms for a number of group
and number theoretic problems that, like Shor's algorithm, use the
quantum Fourier transform as the essential ingredient (e.g.
\cite{ref:Wat01, ref:Kup03, ref:vDhL03, ref:Hal07}), an algorithm
for an oracle graph problem that uses the notion of a quantum random
walk \cite{ref:CCDFGS03}, and recently, algorithms for approximating
combinatorial and topological quantities such as the Jones
Polynomial and the Tutte polynomial \cite{ref:Fre02b, ref:Fre02c,
ref:Aha06b,ref:Aha07}.  These last algorithms related to the Jones
and Tutte Polynomial are fundamentally different from the previous
algorithms. The work presented here began as a study of the core
features of these algorithms.

This work presents a simple new way of looking at quantum
computation.  The consequences are a) new quantum algorithms, b) the
casting of the aforementioned Jones Polynomial and Tutte polynomial
results in a new light, and c) a new geometric view of quantum
computation that will hopefully lead to more new algorithms.
 
The fundamental object for this new view is a {\em tensor network}
which we now briefly describe (a precise description of tensor
networks is given in Section \ref{sec:tnet}).  A tensor network
$T(G, \CM)$, is a graph $G$, a finite set of colors that can be used
to label the edges of $G$, along with a finite array of data $M_v
\in \CM$ assigned to each vertex $v\in V$ of the graph.  This finite
set of data $M_v$ is of the following form: for each possible
coloring of the edges incident to the vertex $v$, the vertex is
assigned a complex number.  Thus for a given coloring $l$ of all
the edges of the graph, each vertex has an assigned value -- we
denote the product of these values by $c_l$. The value of a tensor
network $T(G,\CM)$ is defined to be the sum of $c_l$ over all
possible labelings $l$ of $G$.  
 

The notion of a tensor network geometrically captures many
fundamental linear algebra concepts such as inner product, matrix
multiplication, the trace of a matrix, composition of linear maps,
and the dual space, to mention a few.  Loosely speaking, it is these
features that fundamentally link tensor networks both to quantum
computation and to the various combinatorial and topological objects
discussed in this paper. 

Here, we give a quantum algorithm that takes as input a tensor
network and gives as output an additive approximation of the value
of the tensor network to a certain scale. Together with previous
results \cite{ref:Shi06, ref:Aha06c}, this provides a complete
problem for quantum computation.  We then apply this result to two
classes of problems.

First we give new quantum algorithms for approximating an important
quantity associated to a host of statistical mechanical models.
Statistical mechanical models attempt to model macroscopic behavior
of physical systems made from a very large number of microscopic
systems that interact with each other.  In this paper, we consider a
broad class of models, called q-state models, an example of which is
the well known Potts model.  These models are described by a graph
where the vertices are thought to be in one of q possible states.
For a given assignment of states to the vertices, the energy of the
model is given by a sum of local energy contributions of each edge,
where the local energy is some function of the states of the two
endpoints of the edge.  The partition function of the model is the
sum over all possible assignment of states of the vertices, of a
particular exponential function of the energy of the model for that
assignment (see section \ref{sec:stat} for the details).  It turns
out that many interesting macroscopic properties of the system can
be deduced solely from the partition function. These include the
average energy of the system, its entropy, specific heat, and more
elaborate properties such as phase-transitions \cite{ref:Cal85}. The
calculation of the partition function is therefore an important task
in the theory of statistical physics.  Here we apply the main result
to give an additive approximation of the partition function of any
q-state statistical mechanical model.

Second, we show how a specific application of the main result is
used as an essential step in the recent quantum algorithms for
additively approximating the Jones and Tutte polynomials.

The fact that the approximations are additive and depend on the
approximation scale is by no means a minor point: if the
approximation scale is very large, the algorithm will produce
estimates that are trivial, or at least no better than a classical
approximation.  

So what can be said about the approximation scale in the algorithms
presented in this paper?  In general, of course, our result shows
that there are plenty of problems (i.e., tensor networks) for which
the level of approximation is non-trivial (assuming the power of
quantum computation exceeds that of classical computation.) With
respect to statistical mechanical models, a recent result contained
in \Ref{ref:Nes08} shows that the approximation scale of the
algorithms presented here for certain non-physical statistical
mechanical models is small enough to solve a \BQP -complete problem.
Whether the approximation scale is non-trivial for the statistical
mechanics models with physical parameters remains unknown, though
for a certain range of parameters the scale is superior to any
classical algorithm known to the authors.

Neither tensor networks nor additive approximations are new to the
study of quantum computation.  The fact that highly entangled
quantum states, as well as quantum operations, can be efficiently
represented by tensor networks has been the backbone of many studies
that simulate quantum systems. (see, for example,
Refs.~\cite{ref:Vid03a, ref:Vid04a, ref:Ver04a, ref:Mar05,
ref:Shi06, ref:Nes07a, ref:Vid07a, ref:Hub08}).

Separately, as mentioned, additive approximations and quantum
computation have been linked with the recent results that give
quantum algorithms for additively approximating the Jones polynomial
of braids and the Tutte polynomial of planar graphs, as well as the
complementary results that show that for certain parameters, these
approximations are complete quantum problems \cite{ref:Fre02b,
ref:Fre02c, ref:Aha06b,ref:Aha07}.  Motivated by the Jones
polynomial result, the computational complexity of additive
approximations has been further investigated in \cite{ref:Bor05}.

The view of quantum computation as an additive approximation of
tensor networks provides a useful unifying lens through which to
view these two classes of results.  The above mentioned results
related to classical simulation can be seen as showing that certain
restrictions on the form of a tensor network allow for classical
evaluation.  The algorithmic results for the Jones and Tutte
polynomial, as well as the statistical mechanical algorithms
presented here, can be seen as the quantum approximation of specific
tensor networks whose value is a quantity of interest.

An interesting consequence of the main result presented here is that
two core features of quantum circuits: the unitarity of the gates,
and the notion of time (i.e. that the gates have to be applied in a
particular sequence) are replaced by more flexible features.  The
unitary gate is replaced by an arbitrary linear map encoded in each
tensor.  The notion of time and sequential order of a circuit is
replaced by the geometry of the underlying graph of the tensor
network along with a choice of ``bubbling" of the network, a concept
that is explained in \Sec{sec:algorithm}.  Unlike a quantum circuit,
which is ordered in a unique way, a given tensor-network has many
possible bubblings.  

An outline of this paper is as follows. We begin by giving precise
definitions of tensor networks in \Sec{sec:tnet}, and then in
\Sec{sec:algorithm} we prove the central structural theorem: that
the approximation of tensor networks with the scale prescribed is a
problem a quantum computer can perform efficiently (Theorem
\ref{thm:main}). Section~\ref{sec:complete} then shows that this
approximation problem is a complete problem for quantum computation.
In addition, \Sec{sec:complete} contains a discussion of the
approximation error of this result when applied to particular
families of tensor networks.  We then present quantum algorithms for
approximating the partition function of the statistical mechanics
models in \Sec{sec:stat} (which include the Ising, clock, and Potts
model). \Sec{sec:top} contains a brief discussion of tensor networks
related to some topological invariants. We offer a summary and
discussion in \Sec{sec:summary}.

\section{Preliminaries}
\label{sec:tnet}

\subsection{Notation} 

We fix a $q$ dimensional Hilbert space $H=\BBC^q$ and an orthonormal
basis for $H$ which we denote by $\{\ket{0}, \ket{1}, \dots
\ket{q-1}\}$. For a finite set $S$, $|S|$ will denote the number of
elements of $S$.  A graph $G$ will be denoted by a pair $(V,E)$,
where $V$ is the set of vertices and $E$ is the set of edges.

For any linear operator $A$ over some Hilbert space $H$,  the
norm of $A$ shall mean the operator norm of $A$ and be denoted
by $\norm{A}$. The term $poly(t)$ shall be used to denote some
unspecified polynomial function of $t$.

\subsection{Tensor Networks} 

Tensors are mathematical objects that appear in many branches of 
mathematics and physics.  They can be defined in many ways; for our
purposes, we define them as follows:
\begin{deff}[A Tensor]
\label{def:tensor}

  A tensor $\B{M}$ of rank $k$ and dimension $q$ is an array of
  $q^k$ numbers that are denoted by $M_{i_1, \dots , i_k}$, with
  $i_s$, $1\leq s \leq k$ being indices that take on the values
  $0\leq i_s \leq q-1$. 
\end{deff}
In the rest of the paper, we will always assume that all of the
tensors we deal with are of complex numbers and are tensors of a
fixed dimension $q$.
  
We now describe a couple of useful operations on tensors.  Given a
rank-$k$ tensor $\B{A}$ and a rank-$\ell$ tensor $\B{B}$, the
\emph{product} $\B{A} \otimes \B{B}$ shall be the rank $(k+\ell)$
tensor that is just the tensor product of the two tensors: 
\begin{align*}
  (\B{A}\otimes \B{B})_{i_1 \dots i_k, j_1 \dots j_\ell} 
    \EqDef A_{i_1 \dots i_k}B_{j_1 \dots j_\ell} \ .
\end{align*}

For a rank-$k$ tensor $\B{A}$, and two indices $\ell,m$, $1\leq \ell < m
\leq k$, the \emph{contraction} of $\B{M}$ \emph{with respect to} $\ell$
and $m$ shall be the rank $k-2$ tensor $\B{C}$ given by the
following equation:
\begin{align*}
  \underbrace{C_{i_1 \ldots i_k}}_{\text{no $i_\ell$ and $i_m$}}
    \EqDef \sum_{s=0}^{q-1} 
       A_{i_1 \dots i_{\ell-1},s,i_{\ell+1}
       \dots i_{m-1}, s, i_{m+1}  \dots i_k} \ .
\end{align*}
Combining these two operations together, we can talk about the
contraction of two tensors, which is the result of taking their
product and then contracting the resulting tensor. For example, 
the tensor $C_{i,j} = \sum_{k,\ell} A_{i,k,\ell}B_{j, k, \ell}$ is
the contraction of the two rank-3 tensors $A_{i,k,\ell}$ and $B_{j,
k, \ell}$ along the $k,\ell$ indices. Notice also that
\begin{remark}
  Contraction can be thought of as a generalization of the notions
  of inner product and matrix multiplication. The contraction of two
  rank-1 tensors can be seen as an inner product between two
  vectors.  The matrix product formula $C_{i,j}=\sum_k
  A_{i,k}B_{k,j}$ can be viewed as the contraction of the product of
  two rank-$2$ tensors $\B{A}$ and $\B{B}$ with respect to the
  second index of $\B{A}$ and the first index of $\B{B}$.
\end{remark}

In general, we will be interested in the contraction of the product
of many tensors over multiple indices.  It is an important
observation (that can be easily checked) that the order of products
and contractions does not matter as long as we keep proper track of
the appropriate indices.

This observation leads us to the central object of this paper, the
\emph{tensor network}. This is an extremely useful graphical picture
of tensors, products, and contractions that we now describe. A
rank-$k$ tensor $\B{A}$ shall be represented as a vertex with $k$
edges incident to it -- each edge shall correspond to one index of
$\B{A}$.  The product of two tensors will be represented as the
disjoint union of two such pictures and the contraction of a tensor
along the  $\ell$ and $m$ indices shall be represented by
joining the edge corresponding to the $\ell$'th index with the edge
corresponding to the $m$'th.  With this description, a series of
products and contractions of tensors becomes a graph with labeled
vertices and a certain number of free edges.  The number of free
edges is exactly the rank of the tensor that results from the
products and contractions. Examples of such diagrams are given in 
\Fig{fig:tensors}.
\begin{figure}
  \center \includegraphics[scale=0.7]{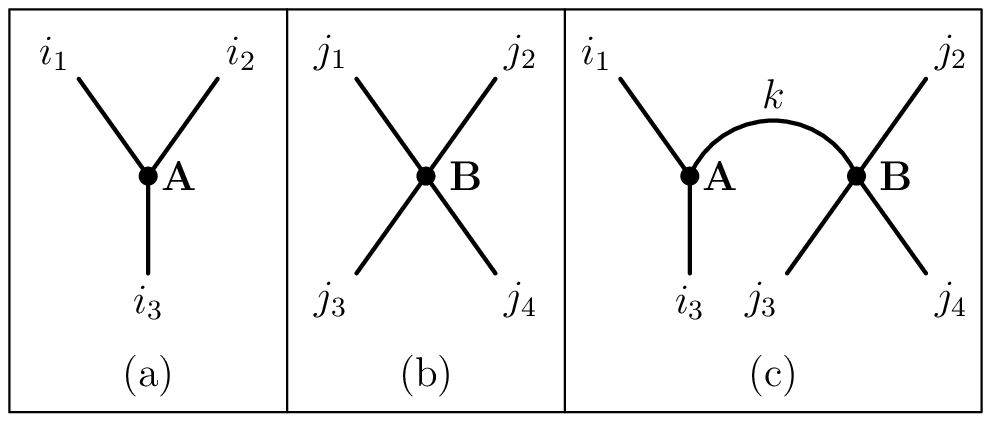} 
  \caption{A graphical representation of tensors: Fig.~(a)
  denotes a rank-3 tensors $A_{i_1, i_2, i_3}$, Fig.~(b) denotes the
  rank-4 tensor $B_{j_1, j_2, j_3, j_4}$, and Fig.~(c) denotes their
  contraction $\sum_k A_{i_1, k, i_3}B_{k, j_2, j_3, j_4}$.
  In the following, we will usually omit the labeling of tensors'
  indices.
  \label{fig:tensors}}
\end{figure}

We shall be particularly interested in cases where all indices are
contracted to yield a single number, or, equivalently, when the
associated graph has no free edges:
\begin{deff}[Tensor network]
  A tensor network is a product of tensors that are contracted
  together such that no free indices are left\footnote{We note that
  in other works in quantum information, it is conventional to
  include also graphs with free edges in the definition of tensor
  networks. Nevertheless, here we narrow the definition to fully
  contracted graphs.} It is denoted by $T(G,\CM)$, with $G=(V,E)$
  being a graph and $\CM=\{\B{M}_v | v\in V\}$ a set of tensors.
  For each $v\in V$, the rank of the tensor $\B{M}_v$ is equal to
  the degree of $v$, with every index of $\B{M}_v$ being associated
  with an adjacent edge of $v$. Finally, each edge denotes a
  contraction of the two indices that correspond to its ends.  The
  \emph{value} of the tensor network is the number that results from
  the series of products and contractions described by the network.
  When the context is clear we shall use $T(G,\CM)$ to denote the
  value of the network as well as the network itself. 
\end{deff}
With these definitions, a tensor network is nicely described
pictorially by a graph, as demonstrated in \Fig{fig:simple-tnet}. 

\begin{figure}
  \center \includegraphics[scale=1]{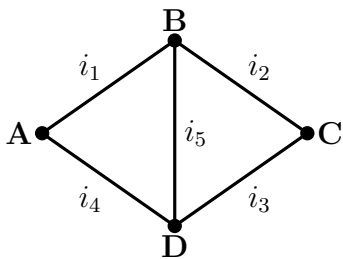} 
  \caption{An example of a simple tensor network
  that is given by $T(G,M) = \sum_{i_1, \ldots, i_5} A_{i_1, i_4}
    B_{i_1, i_5, i_2} C_{i_2, i_3} D_{i_3, i_5, i_4}$.
  \label{fig:simple-tnet}   
  }
\end{figure}

The definition of a tensor network motivates a different notation
that will be especially helpful when studying statistical models.  Given
a tensor network $T(G, \CM)$, we define an \emph{edge labeling $l$}
to be an assignment of an integer $0,1, \ldots, q-1$ to each edge of
$G$. The network tensors can be viewed as \emph{functions} of these
labelings: $M_{v}(l)\EqDef (M_v)_{i_1, \dots , i_k}$ where the value
of the indices $(i_1, i_2, \dots,i_k)$ are defined by the labeling
$l$. With this notation, the value of a tensor network can be neatly
written as a sum over all possible labeling of the edges:
\begin{align}
  \label{eq:alternative}
  T(G,\CM) = \sum_{\text{labeling $l$\ \ }}\prod_{v\in V} M_v(l) \ .
\end{align}

\subsection{Tensors as quantum states and quantum operators}
\label{sec:map}

There is an extremely useful relation between tensors and quantum
states and quantum operators. Given a Hilbert space $H^{\otimes k}$
(recall that $H=\BBC^q$) with a fixed basis $\{\ket{i_1}\otimes
\ket{i_2} \dots \otimes \ket{i_k}\}$, there is a 1-1 mapping between
rank-$k$ tensors of dimension $q$ and vectors in the Hilbert space,
given by:
\begin{align*}
  \B{M}\mapsto \ket{M} \EqDef \sum_{i_1, \dots, i_k=1}^q 
    M_{i_1\dots i_k} \ket{i_1}\otimes \ket{i_2}\otimes
  \dots \otimes \ket{i_k} \ .
\end{align*}

In addition, we can also identify tensors with \emph{linear maps}
from one Hilbert space to another. Given a rank-$n$ tensor $\B{M}$,
we partition the indices of $\B{M}$ into two sets $K$ and $L$.  Set
$k=|K|$, $\ell=|L|$.  Then define $\B{M}^{K,L}:
H^{\otimes k} \rightarrow H^{\otimes \ell}$ to be the map: 
\begin{align}
  \B{M}^{K,L} = \sum_{i_1\dots i_k, j_1 \dots j_\ell} 
    M_{i_1\dots i_\ell, j_1 \dots j_k} 
    \ket{i_1}\otimes\dots \otimes\ket{i_\ell}
    \bra{j_1}\otimes \dots \otimes \bra{j_k},
\end{align}
where in the above sum the $i$ variables range over all possible
values of the indices in $L$ and the $j$ variables range over all
possible values of the indices in $K$.  We further note that
although we wrote $M_{i_1\dots i_\ell, j_1 \dots j_k}$, this will
only be correct if the set $L$ consisted of the first $\ell$ indices
of $\B{M}$ and $L$ the last $k$ indices.  In general the $i$ and $j$
indices will be shuffled around to correspond to the locations of
$K$ and $L$. Alternatively, we can define $\B{M}^{K,L}$ inside the
pictorial world of tensor-networks: it is the linear map that takes
rank-$k$ tensors to rank-$\ell$ tensors. Given a rank-$k$ tensor
$\B{A}$, contract the $k$ indices of $\B{A}$ with the indices $K$ of
$\B{M}$, and the result is a rank-$\ell$ tensor with the indices
$L$. This is demonstrated in \Fig{fig:contraction}.
\begin{figure}
  \center \includegraphics[scale=1.2]{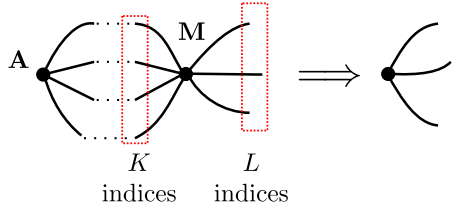} 
  \caption{A graphical illustration of the action of the swallowing
  operator $\B{M}^{K,L}$. It takes a $k$-rank tensor into a
  $\ell$-rank tensor by contracting its $k$ indices with its $K$
  indices. The resultant contraction has $\ell$ indices -- the
  indices that come from the $L$ indices of $\B{M}^{K,L}$.
  \label{fig:contraction}}
\end{figure}

\subsection{Additive approximations} 
\label{sec:add}

The main result of this paper shows that every (finite) tensor
network admits an efficient quantum additive approximation. In this
section we define this type of approximation.

Roughly speaking, an additive approximation algorithm for a quantity
$X$ provides an approximation within the range $[X-\Delta/poly(n),
X+\Delta/poly(n)]$ with $\Delta$ being the approximation
scale and $n$ the running time of the algorithm . The
approximation allows errors up to $\Delta/poly(n)$, whereas a
\emph{multiplicative} approximation only
allows errors up to $|X|/poly(n)$. Since $\Delta$ can be arbitrarily
larger than $|X|$, we have a weaker notion of approximation.
Nevertheless, it appears most suitable in describing the performance
of many quantum algorithms, in particular those which deal with
topological invariants such as the Jones polynomial
\cite{ref:Fre02a, ref:Fre02b, ref:Fre02c, ref:Bor05, ref:Aha06b}. 

In \Ref{ref:Bor05} this type of approximation and its relation to
quantum computation were studied.  We therefore adopt their
definition of this approximation with some minor adjustments.
\begin{deff}[Additive approximation] \label{def:adaprox}
  A function $f:\{0,1\}^*\to \BBC$ has an 
  additive approximation with an approximation scale $\Delta:\{0,1\}^*\to
  \BBR_+$ if there exists a probabilistic
  algorithm that given any instance
  $x\in \{0,1\}^*$ and $\epsilon>0$, produces a complex number $V(x)$ such that
  \begin{align}
    \Pr\Big( |V(x) - f(x)| \ge \epsilon\Delta(x)  \Big) \le \frac{1}{4} \ , 
  \end{align}
  in a running time that is polynomial in $|x|$ and $\epsilon^{-1}$.
\end{deff}
Notice that we did not specify the type of approximation algorithm;
it can be either classical or, as in this paper, quantum.  Note also
that the $1/4$ parameter in the definition can be replaced by
constant $\delta\in (0,1/2)$, since we could reduce this error
probability in polynomial time by taking several runs of the
algorithm. Finally, notice that by setting $\Delta(x)\EqDef |f(x)|$,
we recover the definition of an FPRAS (Fully Polynomial Randomized
Approximation Scheme).

As noted in the introduction, additive approximation can be trivial
if the approximation scale $\Delta$ is too large. In the quantum
case, the approximation scale might be non-trivial, yet it might be
classically reproducible. We discuss this problem in
\Sec{sec:complete}.


\section{Approximating a tensor network with a quantum computer}
\label{sec:algorithm}

\subsection{Outline of the algorithm}

We now describe the quantum algorithm that gives an additive
approximation of a tensor network. We begin with an informal
pictorial description of the algorithm. Given a tensor network, we
imagine its graph embedded in $\BBR^3$, and a large bubble that
approaches the graph and starts ``swallowing'' it one vertex at a
time, as illustrated in \Fig{fig:bubbeling-1}. Every time it
swallows a vertex, it also swallows some of its adjacent edges,
while the rest of the adjacent edges are only ``half swallowed'';
they intersect with the boundary of the bubble. Thus between
swallows, what remains are the vertices that we have not yet
swallowed, the edges between those vertices, as well as half-edges
that join swallowed vertices with un-swallowed ones.  In the end,
once all vertices have been swallowed, we are left with a rank-zero
tensor that is simply a number -- the value of the tensor-network.  

Our algorithm will mimic this swallowing step by step by creating a
state related to the tensor of the swallowed part of the graph. The
act of swallowing a new vertex will be mirrored by an application of
a ``swallowing" operator.  The swallowing of a graph can be
therefore viewed as a generalization of a standard quantum
computation, with the vertices being the gates, and the bubbling
determines the order in which these gates are applied. 

There is one obstacle, however, that one has to cross: the
swallowing operators are linear, but not necessarily unitary. In
order to implement them on a quantum computer, we will use a
standard trick: we will simulate the non-unitary operator as a
sub-unitary operator. This is done by adding an ancilla qubit to the
system, followed by a global unitary operation, and then a
projection on a specific subspace of the ancilla. The resultant
state will be a scaled version of the vector we would have obtained
by the non-unitary transformation, with the scaling factor being the
norm of the non-unitary operator. As we shall see, this will result
in an approximation scale that is the product of the norms of the
swallowing operators. 

\begin{figure}
  \center \includegraphics[scale=0.7]{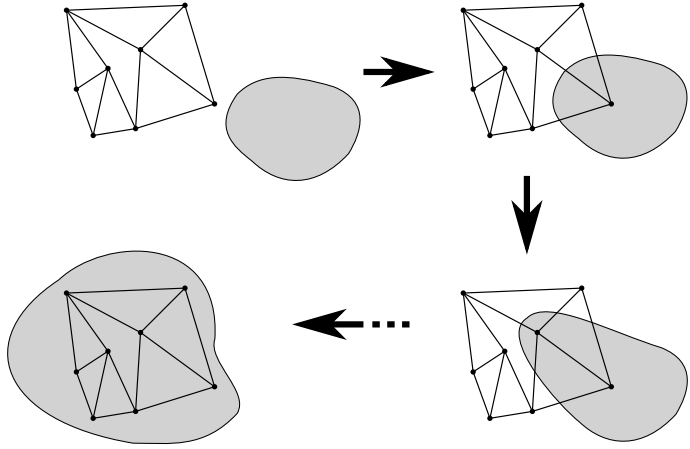} 
  \caption{An illustration of a bubbling of a tensor network with 8
  vertices. In the beginning, the bubble is away from the graph, but
  then it starts swallowing its vertices one by one. The process
  ends when all the vertices are inside the bubble.
  \label{fig:bubbeling-1}   
  }
\end{figure}

\subsection{Bubbling and swallowing operators}
\label{sec:bubbling}

We begin with formal definitions of a bubbling and the resultant
swallowing operators discussed above.
\begin{deff}[Bubbling of a graph]
  A bubbling $B$ of a graph $G=(V,E)$ shall mean an ordering of all
  the vertices of $G$, 
  \begin{align}
    v_1, v_2, v_3, \ldots  \ .  
  \end{align}
  This ordering induces a sequence of  subsets 
  \begin{align}
    \emptyset=S_0\subset S_1\subset S_2 \subset\dots \subset S_{|V|}=V ,
  \end{align}
  with $S_i=\{v_1,\ldots, v_i\}$.  For each $i$, we define
  $Z_i\subset E$ to be the set of edges with exactly one endpoint in
  $S_i$. A graphical illustration of this process is shown in
  \Fig{fig:bubbeling-1} and \Fig{fig:bubbeling2}.
\end{deff} 

\begin{figure}
  \center \includegraphics[scale=1]{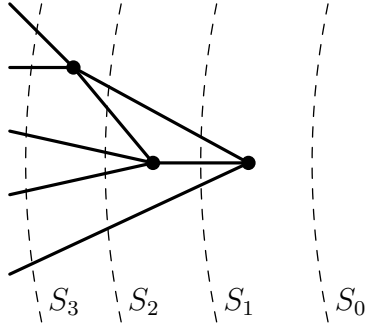} 
  \caption{The first 4 stages of a bubbling of a graph.
  \label{fig:bubbeling2}}
\end{figure}

Given a tensor network $T(G, \CM)$, a bubbling $B$ of $G$ also
defines a sequence $\B{A}_i$, $0\leq i \leq n$ of $n+1$ tensors as
follows. For every $i$, cut the tensor network at the edges in
$Z_i$; this divides the network into two pieces, one piece contains
all the vertices of $S_i$, the other contains the remaining
vertices. Define $\B{A}_i$ to be the rank $|Z_i|$ tensor represented
by the first piece of the dissected graph (this corresponds to the
tensor of the swallowed part of the graph in our informal
description). For the special $i=0$ case, $\B{A}_0$ has no indices;
it is a single number, which we define to be $1$.  The last tensor,
$\B{A}_n$, is also a zero-rank tensor: it corresponds to the
contraction of the entire network, hence its value is $T(G,\CM)$ --
the number we are trying to approximate.

The relationship between $\B{A}_i$ and $\B{A}_{i+1}$ is clear:
$\B{A}_{i+1}$ is obtained from $\B{A}_i$ by contracting it with
$\B{M}_{v_i}$ over the indices in $Z_i$. This action is familiar, it
is just the application of the map $\B{M}_{v_i}^{I,J}$ with $I=Z_i$
contracted with the corresponding edges of $\B{A}_{i}$ in $Z_i$.
This leads us to the following definition:

\begin{deff}[The swallowing operator]
\label{def:swallow-op}
  Let $T(G,\CM)$ be a tensor network with a bubbling $B=(v_1,v_2,
  \dots v_n)$.  For every integer $i \in \{1,\dots n\}$ define:
  \begin{itemize}
    \item $K$ to be the set of \emph{input edges} -- all edges in 
          $Z_{i-1}$ that are connected to $v_i$. These edges
          connect $v_i$ to $S_{i-1}$.
          
    \item $L$ to be the set of \emph{output edges} -- all edges in 
      $Z_i$ that are connected to $v_i$. These edges connect
      $v_i$ to $V\setminus S_i$.

    \item $J$ to be the set of \emph{untouched edges} -- edges in 
          $Z_{i-1}$ that are not adjacent to $v_i$ (these edges
          must also be in $Z_{i}$).
  \end{itemize}
   
  We define the \emph{swallowing operator} $\CO_{v_i}$ to be a
  linear operator that takes states from $H^{\otimes|Z_{i-1}|}$ to
  $H^{\otimes|Z_i|}$ by
  \begin{align}
  \label{def:Ov}
    \CO_{v_i} = \Id_{J}\otimes \B{M}_{v_i}^{K,L} \ ,
  \end{align}
  where by $\Id_{J}$ we mean the identity operator on the indices
  corresponding to the untouched edges of $J$.
\end{deff}
With this definition, it is clear that 
\begin{align} 
\label{eq:swallow}
  \ket{A_{i+1}}= \CO_{v_i} \ket{A_i} \ , 
\end{align}
and we are ready to state the formal approximation algorithm in the
next section.

\subsection{Implementation on a quantum computer}

As mentioned in the beginning of the section, the main technical
difficulty in implementing the algorithm on a quantum computer is
the fact that the swallowing operators might be non-unitary. The
following lemma serves as the central building block of
the algorithm. It shows the well-known result that any such operator
can be implemented on a quantum computer using ancilla qubits,
unitary operators and a final projection.
\begin{lemma} 
\label{lem:tounitary} 
  Given a linear map $A:H^{\otimes k} \rightarrow H^{\otimes k}$,
  let $\mathbbm{B}$ denote the space corresponding to a qubit with
  computational basis $\{\ket{0}, \ket{1} \}$.  Then there exists a
  unitary operator $U:H^{\otimes k} \otimes \mathbbm{B}
  \rightarrow H^{\otimes k} \otimes \mathbbm{B} $ such that
  \begin{align} 
    \label{eq:U-action}
    U\big( \ket{\alpha} \otimes \ket{0}\big) 
      = \frac{1}{\norm{A}} \big(A\ket{\alpha}\big)\otimes \ket{0} 
        + \ket{\beta_2}\otimes \ket{1} \ . 
  \end{align}
  Furthermore, $U$ can be implemented on a quantum computer in time
  $\poly(q^k)$ with exponential accuracy. (where $q$ is the
  dimension of $H$).
\end{lemma}

\begin{proof}
  Set $m=q^k$.  Using the fact that every linear operator has a
  singular value decomposition, we can write
  $\frac{1}{\norm{A}}A=V_1DV_2$ where $V_1$ and $V_2$ are unitaries 
  and $D$ is a diagonal matrix with diagonal entries 
  \begin{align}
    1\ge r_1\ge r_2\ge\ldots\ge r_m \ge 0 \ .
  \end{align}
  Now define the map $U_D: H^{\otimes k} \otimes \mathbbm{B}
  \rightarrow H^{\otimes k} \otimes \mathbbm{B}$ as follows: every
  vector in $H^{\otimes k} \otimes \mathbbm{B}$ can be written as
  unique super position 
  \begin{align}
    \ket{\alpha} = \ket{\beta _0} \otimes \ket{0} + \ket{\beta _1}
        \otimes \ket{1} \ .
  \end{align}
  Then action of $U_D$ on $\ket{\alpha}$ is defined by
  \begin{align}
    U_D\ket{\alpha} \EqDef \big(D\ket{\beta_0} +
        \sqrt{1-D^2}\ket{\beta _1}\big)\otimes\ket{0} 
    + ( -\sqrt{1-D^2}\ket{\beta _0} + D \ket{\beta _1} )
  \otimes \ket{1} \ , 
  \end{align}
  where $\sqrt{1-D^2}$ is the diagonal matrix with $i$'th entry
  $\sqrt{1-r_i^2}$.  It is a simple calculation to verify that $U_D$
  is unitary.  Setting $U\EqDef \big(V_1\otimes
  \Id_{\mathbbm{B}}\big) U_D \big(V_2\otimes
  \Id_{\mathbbm{B}}\big)$, \Eq{eq:U-action} follows.

  It remains to show that the operation of computing and applying
  $U$ can be done in quantum $\poly(m)$ time (recall $m=q^k$). The
  computation of the singular value decomposition can be done in
  classical $\poly(m)$ time.  We are left to implement the three
  unitaries $V_1\otimes \Id$, $U_D$ and $V_2\otimes \Id$.  Since
  these operators act on an $2m$ dimensional space they are
  therefore $(\log m)$-local qubit operators.  The simulation of
  unitary operators that act on $\log m$ qubits is a standard
  procedure in quantum computation based on the Solovay-Kitaev theorem and can be done in $\poly(m)$
  quantum time \cite{kityushen,ref:Daw05} and thus the whole process can be
  completed in $\poly(m)$ time.
\end{proof}

With this lemma in hand, we are ready to prove the central result of
this paper:
\begin{theorem}[Additive quantum approximation of a tensor-network]
\label{thm:main} 
  Let $G=(V,E)$ be a graph of maximal degree $d$, and let $T(G,\CM)$
  be a tensor-network of dimension $q$ defined on $G$. For a given
  bubbeling 
  $B=(v_1, v_2, v_3, \ldots)$ of $G$, let $\big\{\CO_{v}\big\}_{v\in
  V}$ the
  corresponding swallowing operators from \Def{def:swallow-op}. Then for
  any error parameter $\epsilon>0$, there exists a quantum algorithm
  that runs in $|V|\cdot\epsilon^{-2}\cdot\poly(q^d)$ quantum time
  and outputs a complex number $r$, such that
  \begin{align}
    \Pr\Big(|T(G,\CM)-r| \ge \epsilon \Delta\Big) \le \frac{1}{4} \ ,
  \end{align}
  with 
  \begin{align}
  \label{eq:approx}
    \Delta \EqDef \prod_{v\in V}\norm{\CO_v} \ .
  \end{align}
\end{theorem}

\begin{proof}
  
 We set $n\EqDef |V|$.
  As noted the end of \Sec{sec:bubbling}, given a tensor network and
  a bubbeling $B=(v_1, \ldots, v_n)$, the process of swallowing
  defines a series of vectors that live in possibly different
  Hilbert spaces. We start with a normalized vector $\ket{\Omega}$
  that lives in a one-dimensional Hilbert space. Then
  \begin{align}
    \ket{A_1} &= \CO_{v_1}\ket{\Omega} \ , \nonumber \\
    \ket{A_2} &= \CO_{v_2}\ket{A_1} = \CO_{v_2}\CO_{v_1}\ket{\Omega} \ , 
      \nonumber \\
    \vdots  \label{eq:A-vec}\\
    \ket{A_n} &= \CO_{v_n}\cdots\CO_{v_1}\ket{\Omega} =
      T(G,M)\ket{\Omega} \ .\nonumber
  \end{align}
  The above states, of course, cannot be directly stored in a
  quantum computer because they are not necessarily normalized. We
  solve this problem by moving to a larger Hilbert space by adding
  ancilla qubits. Specifically, at the $i$'th step, the state of
  swallowing process will be stored in a \emph{normalized} state
  $\ket{\psi_i}$ that lives in the Hilbert space
  $H^{\otimes|Z_i|}\otimes \BBB^{\otimes i}$.
  $H^{\otimes|Z_i|}$ is the space where $\ket{A_i}$ lives
  (see \Sec{sec:bubbling}), and $\BBB^{\otimes i}$ is the space of
  $i$ ancilla qubits. Any vector in that space can be uniquely
  expanded in the standard basis of the ancilla qubits:
  \begin{align}
    \ket{\phi} = \sum_{s} \ket{v_s}\otimes \ket{s} \ .
  \end{align}
  Here $\ket{v_s} \in H^{\otimes|Z_i|}$, and $\ket{s}$ is a standard
  basis element of the ancilla subspace that corresponds to the
  classical string $s$ of $i$ bits. In terms of this expansion, the
  state $\ket{\psi_i}$, which encapsulates the state of the
  swallowing process, will be given by
  \begin{align}
  \label{eq:psi-i}
    \ket{\psi_i} = 
      \frac{1}{\norm{\CO_{v_1}}\cdots\norm{\CO_{v_i}}}\ket{A_i}
         \otimes
        \underbrace{\ket{0}\otimes\cdots\otimes\ket{0}}_{\text{$i$
        ancilla qubits}}
        + \ldots
  \end{align}
  In other words, projecting the ancilla qubits of $\ket{\psi_i}$ on
  $\ket{0}\otimes\cdots\otimes\ket{0}$, will result in the state 
  \begin{align*}
    \frac{\ket{A_i}}{\norm{\CO_{v_1}}\cdots \norm{\CO_{v_2}}} \ .
  \end{align*} 
  Notice that by Eqs.~(\ref{eq:A-vec}), the norm of this vector is
  necessarily smaller than or equal to 1.

  Let now show that we can efficiently generate the $\ket{\psi_i}$
  states in \Eq{eq:psi-i} using a quantum computer. The proof is by
  induction.  We start by generating $\ket{\psi_1}$. Denote by $k$
  the degree of the vertex $v_1$. Since this is the first vertex to
  be swallowed, $\CO_{v_1}$ has no input edges and exactly $k$
  output edges. Its domain is therefore a trivial one dimensional
  space, which is spanned by a normalized vector $\ket{\Omega}$, and
  its co-domain is $H^{\otimes k}$.
  
  Then we define the operator $\tilde{\CO}_{v_1}:H^{\otimes k}\to
  H^{\otimes k}$ by 
  \begin{align}
    \tilde{\CO}_{v_1}\big(\ket{0} \otimes \dots \otimes \ket{0}\big) \EqDef
      \CO_{v_1}\ket{\Omega}= \ket{A_1}
  \end{align} 
  and $\tilde{\CO}_{v_1} \ket{\alpha} =0$ for every $\ket{\alpha}\in
  H^{\otimes k}$ that is orthogonal to $\ket{0} \otimes \dots
  \otimes \ket{0}$. Notice that $\norm{\tilde{\CO}_{v_1}} =
  \norm{\CO_{v_1}}$. We initialize one extra ancilla qubit to
  $\ket{0}$ and apply Lemma~\ref{lem:tounitary} to the all zero
  basis state to get
  \begin{align}
    \ket{\psi_1} =  \frac{1}{\norm{\CO_{v_1}}}\ket{A_1}\otimes\ket{0}
    + \ldots
  \end{align}
  Here, as in \Eq{eq:psi-i}, the 3 dots stand for the rest of the
  terms that one obtain when decomposing $\ket{\psi_1}$ according to
  the standard basis of the ancialla qubit (in other words, here it
  would be some vector $\ket{\phi}\otimes\ket{1}$). Notice also that the
  whole process can be done in $\poly(q^k)\le \poly(q^d)$ quantum
  time. This proves the $i=1$ case.
  
  Assume now that we have ``swallowed'' $i-1$ vertices, and have
  created the state $\ket{\psi_{i-1}}$. By \Def{def:swallow-op},
  $\CO_{v_i}=\Id_J\otimes M^{K,L}$, with $K$ corresponding to the
  set of input edges, $L$ to the set of output edges, and $J$ to the§
  set of untouched edges that are not connected to $v_i$.  We can
  therefore ignore all the registers that correspond to the
  untouched edges and concentrate only on the ``active'' registers
  $K,L$. These are at most $d$ registers, each holding numbers between
  $0,\ldots ,q-1$. All together they are therefore described by
  mostly $d\log q$ qubits.
  
  Define $k=|K|$ and $\ell=|L|$, (note that $k+\ell \leq d$), and
  assume first that $k=\ell$. Then we initialize one new qubit to
  $\ket{0}$ and apply Lemma~\ref{lem:tounitary} to $M^{K,L}$ (which
  is now a square matrix), and transform $\ket{\psi_{i-1}}$ into
  $\ket{\psi_{i}}$ with the property that
  \begin{align}
  \label{eq:temp-state}
  \ket{\psi_i} = 
    \frac{1}{\norm{M^{K,L}}} \big[ \Id_J \otimes M^{K,L}
      \text{$\ket{A_{i-1}}$}
      \big]\otimes
      \underbrace{ \ket{0} \otimes 
        \dots \otimes \ket{0}}_{\mbox{$i$ qubits}} + \ldots
  \end{align}
  
  Using the fact that $\norm{M^{K,L}}=\norm{\CO_{v_i}}$, and that
  $\big( \Id_J \otimes M^{K,L}\big)\ket{A_{i-1}}=
  \CO_{v_i}\ket{A_{i-1}}=\ket{A_i}$, we find that $\ket{\psi _i}$
  satisfies (\ref{eq:psi-i}). Moreover, this transformation is done
  in $\poly(q^k)\le \poly(q^d)$ time.
  
  When $k<\ell$ we simply add $\ell-k$ input registers and set them
  to $\ket{0}$.  We redefine $M^{K,L}$ to be a square matrix that
  operates identically as the original transformation, provided that
  the extra $k-l$ registers are all set to $\ket{0}$, and otherwise
  acts as the zero operator. This guarantees that the $M^{K,L}$
  preserves its norm. We can now repeat the $k=\ell$ case. Notice
  that this process can also  be done in $\poly(q^\ell)\le \poly(q^d)$
  time.
  
  Similarly, in the $k>\ell$ case, we add $k-\ell$ output registers
  and redefine $M^{K,L}$ to always set them to $\ket{0}$, hence
  preserving its norm. This finishes the proof of the induction.
  
  Repeating this process all the way up to $i=n$, we generate in
  $n\cdot\poly(q^d)$ time the state
  \begin{align}
    \ket{\psi_n} = \frac{T(G,\CM)}{\prod_i\norm{\CO_{v_i}}}
    \underbrace{\ket{0}\otimes\dots\otimes\ket{0}}_{\text{$n$ ancilla qubits}}
    + \ldots \ ,
  \end{align}
  which lies in a Hilbert space that is composed entirely of ancilla
  qubits.
  We now use the Hadamard test (see \App{sec:Hadamard}) to estimate
  the inner product of $\ket{\psi_n}$ with
  $\ket{0}\otimes\dots\otimes\ket{0}$. For any given $\epsilon>0$,
  we generate $\CO(\epsilon^{-2})$ copies of $\ket{\psi_n}$ and
  after the appropriate measurement, we obtain a complex number $r'$
  such that 
  \begin{align}
    \Pr\Big[\Big| r' - T(G,\CM) / \prod_i\norm{\CO_{v_i}}\Big| \ge
    \epsilon \Big] \le \frac{1}{4} \ .
  \end{align}
  All together, this is done in $|V|\cdot\epsilon^{-2}\cdot\poly(q^d)$
  time. Multiplying by $\Delta=\prod_i\norm{\CO_{v_i}}$, and
  outputing $r=\Delta r'$, proves the theorem.

\end{proof}

We note that $\Delta$, the additive scale of approximation of a
given tensor network, depends on the choice of bubbling.  The issue
of finding a good bubbling for a given tensor network is far from
trivial. In fact, it is known that the closely-related problem of
finding the tree width \cite{ref:Mar05} of an arbitrary
graph\footnote{The tree width of a graph is equivalent to the best
bubble width of a graph\cite{ref:Aha06c}. For a given bubbling, the
bubble width is the maxmial number of edges crossing the bubble
during the swallowing process. The best bubble width is therefore
the minimal bubble width over all bubblings of the
graph.}\label{footnote:bubble-width} is NP-hard \cite{ref:Arn87}. In
the present work we treat the bubbling as an external object, given
to us together with the tensor-network. In \Sec{sec:stat}, where we
construct tensor-networks that calculate the partition functions of
classical models, we employ a simple bubbling that seems to have
certain advantages, but in no way does this mean that these are
optimal. Nevertheless, already from these examples it becomes clear
how different bubblings can yield very different approximation
scales for the same network.


\section{The hardness and completeness of approximating a tensor network}
\label{sec:complete}

As mentioned earlier, special attention needs to be paid to the
nature of the additive approximation in Theorem~\ref{thm:main}.
Additive approximations are tricky; the approximation scale $\Delta$
might be exponentially larger than $|T(G,\CM)|$, in which case the
output of the algorithm is meaningless. Unfortunately, ruling out
this possibility is difficult: if we want to bound the ratio between
$\Delta$ and $|T(G,\CM)|$, we must have some other, external
estimate for the latter, which might generally be a difficult task.

We address this issue in a two ways. First, we will consider
the \emph{hardness} of the approximation. We will show that our
additive approximation is BQP-hard for certain classes of networks, which
will prove that the approximation is non-trivial for many instances
of the problem. 

Second, we will focus on the
\emph{completeness} of the approximation. We will argue that
Theorem~\ref{thm:main} can be viewed as a new way of casting
quantum computation rather than as an approximation result.

We begin with a brief review of known complexity results for the
evaluation of tensor networks. 

\subsection{Classical hardness results}

Tensor-networks can be very hard to evaluate; they are
the sum an exponential number of terms. A well-known example is the
3-coloring problem of a graph \cite{ref:Pap94}. We say that a graph
$G$ is 3-colorable if it is possible to color its vertices in one of
3 colors such that adjacent vertices would always be colored
differently. Deciding whether a graph is 3-colorable or not is a
famous \NP-hard problem, even when restricted to the class of planar
graphs of degree 4 \cite{ref:Gar74}. Moreover, counting the number
of possible colorings in known to be a \SharpP-complete problem
(see, for example, \Ref{ref:Dye04}). However, given a graph
$G=(V,E)$, it is an easy exercise to construct a tensor network that
counts its total number of 3-colorings: set $q=3$, define the
tensors at the vertices to give $1$ if all their edges are colored
identically and zero otherwise, and finally place new vertices in
the middle of each edge, and define its tensor to give $0$ when its
two edges are labeled with the same color and $1$ otherwise. We leave it to the reader
to verify the correctness of this construction. 

So exact evaluation of a tensor-network might be \SharpP-hard. But
what about approximations, in particular multiplicative
approximation?  As we allude to in \Sec{sec:top}, one can define a
tensor-network that calculates the \emph{Tutte polynomial} of a
planar graph. This is a two-variable polynomial that can be defined
for every graph $G$. It encodes an extremely wide range of
interesting combinatorial properties of $G$, making it central in
graph theory \cite[p.~45]{ref:Wel93}. Its exact evaluation turns out
to be \SharpP-hard at all but trivial points \cite{ref:Jae90}.
Moreover, recent results show that even a multiplicative
approximation to it (FPRAS)\footnote{See the paragraph following Definition  \ref{def:adaprox} for a definition of an FPRAS approximation.} is \NP-hard, (and sometimes even \SharpP-hard) for a
large part of the Tutte plane \cite{ref:Les07}. Therefore there
exists families of tensor-network for which FPRAS approximation is
also \NP-hard.

Finally, it turns out that additive approximations can also be
\NP-hard. Indeed, a simple construction in Theorem~4.4 of
\Ref{ref:Bor05} shows that an additive approximation of the
$q$-coloring problem with a scale $\Delta=(q-1-\delta)^{|V|}$, for
any $0<\delta<q-1$ and $q\ge 3$ is \NP-hard since it can be used to
decide whether a graph is q-colorable.

\subsection{Quantum hardness results}

We now show that there exist classes of tensor-networks for
which additive approximations are \BQP\ hard.  As
we have seen in \Sec{sec:map}, tensor networks can represent quantum
states, as well as linear maps over these states.  It is not
surprising that quantum circuits can be represented by tensor
networks.  Indeed, this observation appears in many recent studies
that try to draw the border between quantum and classical
complexities, characterize the nature of entanglement, and find
efficient algorithms to simulate certain classes of quantum systems
(see, for example, Refs.~\cite{ref:Vid03a, ref:Vid04a, ref:Ver04a,
ref:Mar05, ref:Shi06,ref:Aha06c,ref:Vid07a, ref:Hub08}). For sake of
completeness, we show how this encoding can be done, but see
\Ref{ref:Mar05} for a broader view.

Consider a quantum circuit $Q=Q_{L}\cdot\ldots\cdot Q_1$ that is
defined on $n$ qubits.  Denote by $p_0$ the probability of measuring
a $0$ in the last qubit of the original circuit $Q$ applied to
$\ket{0^{\otimes n}}$.  To perform universal quantum computation, it
is enough to distinguish between the cases when $p_0< 1/3$ and
$p_0>2/3$ for any circuit $Q$.  We define a related circuit $U$ on
$n+1$ qubits: $U$ applies $Q$ to $\ket{0^{\otimes n}}\EqDef\ket{0}^{\otimes n}$, then copies
the last qubit of $Q$ to the additional qubit by a $CNOT$ gate, and
then applies $Q^{-1}$ (see \Fig{fig:Q-circuit}).  It is a
straightforward algebraic exercise to show that $\bra{0^{\otimes
n}}U\ket{0^{\otimes n}}=p_0$ for the original circuit $Q$.  We will
construct a tensor network whose value is $T(G,\CM)=\bra{0^{\otimes
n}}U\ket{0^{\otimes n}}$ such that a straightforward bubbling of
this network, which is associated with the original ordering of the
circuit, will yield an approximation scale $\Delta=1$. This will
enable us to distinguish between the two cases $\bra{0^{\otimes
n}}U\ket{0^{\otimes n}}=p_0\ge 2/3$ or $\bra{0^{\otimes
n}}U\ket{0^{\otimes n}}=p_0\le 1/3$ and thus show that the problem
of additively approximating tensor networks to the scale described
is quantum hard.

\begin{figure}
  \center \includegraphics[scale=0.7]{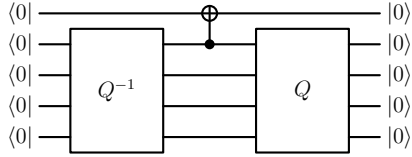} 
  \caption{Constructing a quantum circuit $U$ from another circuit
  $Q$ such that $\bra{0^{\otimes n}}U\ket{0^{\otimes n}}$ is equal
  to the probability measuring the first qubit of $Q\ket{0^{\otimes n}}$ in
  the state $\ket{0}$.
  \label{fig:Q-circuit}}
\end{figure}

Let us now define the tensor-network. The dimension of every tensor
is $q=2$, corresponding to the the two possible values of a qubit.
The network consists of 3 types of tensors:
\begin{itemize}
  \item Every $d$-local gate $Q$ is translated into a $2d$-rank tensor
    with $d$ input edges and $d$ output edges:
    \begin{align}
      \label{eq:gate-tensor}
       M^{(Q)}_{k_1, \ldots, k_d; \ell_i, \ldots, \ell_d} \EqDef
       \bra{\ell_1}\otimes\ldots\otimes\bra{\ell_d} Q
       \ket{k_1}\otimes\ldots\otimes\ket{k_d} \Longleftrightarrow
       \includegraphics[bb=-25 -3 25 10] {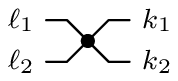}       
    \end{align}
   
  \item Every input qubit $\ket{0}$ is translated into a rank-1 tensor
    \begin{align}
      \label{eq:ket-tensor}
      M^{\ket{0}}_{k} = \delta_{k,0} \Longleftrightarrow
       \includegraphics{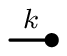}
    \end{align}
    
  \item Every output qubit $\bra{0}$ is translated into a rank-1 tensor
    \begin{align}
      \label{eq:bra-tensor}
      M^{\bra{0}}_{\ell} = \delta_{\ell,0} \Longleftrightarrow
       \includegraphics{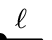}
    \end{align}

\end{itemize}
Contracting these tensors according to the topological structure of
the circuit, we obtain a tensor network $T(G,\CM)$, and it is a
straightforward exercise to check that $T(G,\CM) = \bra{0^{\otimes
n}}U\ket{0^{\otimes n}}$. Finally, when bubbling the network
according to the natural evolution of the circuit, the swallowing
operators of the tensors associated to the gates become the gates
themselves, hence their norm is 1. Similarly, the norms of the
swallowing operators that are associated with
Eqs.~(\ref{eq:ket-tensor}, \ref{eq:bra-tensor}), are also easily
seen to be 1. All in all we have an approximation scale $\Delta=1$.

We therefore reach to the following corollary:
\begin{corollary}
  There exist families of tensor networks for which the additive
  approximation in Theorem~\ref{thm:main} is
  $\mathsf{BQP}$-hard. In particular, all families that
  correspond to families of universal quantum circuits by the
  construction above.
\end{corollary}

It is therefore evident that, for certain collections of tensor
networks, the approximation in Theorem~\ref{thm:main} is
non-trivial. This does not necessarily hold for every member in such
``universal'' families of networks, only for the family as a
whole. The universal families cited in the above corollary originate
from quantum circuits and the quantum universality of their
approximation relies heavily on the unitarity of their operators.
There are, however, other universal families of tensor networks that
are not so tightly related to quantum computation. In \Sec{sec:top},
we refer to one of these families, a family of tensor networks that
approximate the multivariate Tutte polynomial. Unlike the example
above, their underlying operatorial structure is non-unitary. The
proof that approximating these tensor-networks is quantum-hard can
be found in \Ref{ref:Aha07}.

The fact that quantum circuits can be viewed as tensor networks is
the main theme in the paper of Markov \& Shi \cite{ref:Mar05} and
later in \Ref{ref:Aha06c}. These papers study the question of when a
tensor-network can be evaluated classically. \Ref{ref:Mar05} uses
the notion of \emph{tree width} of a graph, which is equivalent to
the notion of \emph{bubble width} that is used in \Ref{ref:Aha06c}
(see footnote on page~\pageref{footnote:bubble-width}). It is shown
that a sufficient condition for an efficient evaluation of a
tensor-network is that the tree width (or, equivalently, its bubble
width) of the graph is of logarithmic size.  To minimize the running
time of the simulation, one should choose a bubbling with a minimal
bubble width. This is done \emph{regardless} of the original
ordering of the circuits. This leads us to the reconsider
Theorem~\ref{thm:main} as essentially a new view of quantum
computation, which we discuss in the next section.

\subsection{Completeness: tensor networks as a different point of view on quantum
  computation - the role of interference}
\label{sec:QC}

Loosely speaking, a useful quantum algorithm must manipulate the
\emph{interferences} of the wave function in a smart way; an
instance of a YES/NO problem must be mapped to a circuit that
produces a constructive/destructive interference. Formally, a
quantum algorithm solves a decision problem, if for every instance
$x$, we can (efficiently) generate a quantum circuit $U_x$ such that
$\bra{0^{\otimes n}}U_x\ket{0^{\otimes n}} \leq 1/3$ for a YES
instance and $\bra{0^{\otimes n}}U_x\ket{0^{\otimes n}} \ge 2/3$ for
a NO instance.

In view of Theorem~\ref{thm:main}, we can rephrase this demand in
terms of tensor-networks. The constructive interference demand
translates into $|T(G,\CM)|$ being of the same order of the
approximation scale $\Delta$ (or, only polynomially smaller). In
other words, \emph{we care less about the actual value $T(G,\CM)$ of
the tensor network and more about the ratio
$|T_x(G,\CM)|/\Delta_x$}. The following corollary can be seen as an
alternative formulation for an efficient quantum computation that
stresses this point:
\begin{corol}[Efficient tensor-network based quantum computation]
\label{corol:quant-comp} 

  A decision problem is in BQP if and only if there is an efficient
  classical transformation that maps any instance $x$ of the problem
  into a tensor-network $T_x(G_x,M_x)$ over a fixed $q$ and a graph
  $G_x$ of a maximal degree $d$, and a corresponding bubbling $B_x$
  such that:
  \begin{align}
    \text{$x$ is a `YES' instance} \quad&\Rightarrow\quad \frac{|T_x(G,\CM)|}{\Delta_x} 
      \ge 2/\poly(|x|)  \ , \\
    \text{$x$ is a `NO' instance} \quad&\Rightarrow\quad \frac{|T_x(G,\CM)|}{\Delta_x} 
      \le 1/\poly(|x|) \ ,
  \end{align}
  where $\Delta_x$ is as in Theorem \ref{thm:main}.
\end{corol}

This formulation generalizes our notion of quantum
computation in two ways: 
\begin{itemize}
  \item Time evolution of the circuit is no longer fixed, we may pick any 
    bubbling that provides the best approximation scale.
    
  \item Unitarity is gone. The operators no longer need to be 
    unitary or even have the same domain and co-domain. We are free
    to construct circuits with non-unitary gates, as well as graphs
    with arbitrary topology, as long as the vertex degree is
    bounded. The approximation scale, however, might no
    longer be $1$.
\end{itemize}
We are hopeful that this new way of looking at quantum computation
will lead to new algorithms beyond those discussed in the next two sections.

\section{Classical statistical mechanics models}
\label{sec:stat}

In this section we present a set of models from statistical physics
that can be defined on arbitrary graphs. We will see how the
tensor-network formalism of the previous section can be used to
construct efficient quantum algorithms that approximate the
partition function of these models.  After introducing these models, we will construct the corresponding tensor network and apply our results to produce a general quantum algorithm (Corollary \Ref{}).  Though these
algorithms are new, the connection between classical statistical
mechanics models and quantum computation is not, and has been
discussed previously in \Ref{ref:Nes06, ref:Nes07,ref:Comm07,
ref:Aha07, ref:GerLid08, ref:Ger08}. We hope that this section will
enrich and clarify the nature of this connection.

A proper introduction of these statistical models is far beyond the
scope of this paper; here we will only provide the details necessary
for understanding the tensor-network constructions.  An interested
reader can find an introduction to these models in any standard text
book on the subject, for example, \Ref{ref:Cal85}.  

%
%

\subsection{A brief introduction to classical statistical models on
  graphs}
\label{sec:stat:intro}

In statistical physics, one is often interested in the macroscopic
behavior of a system that is made from a very large number of
microscopic systems that interact with each other. In
most cases, the everyday systems that we wish to describe are far
too complex to be treated analytically. A common practice is
therefore to study toy models, which are simple enough to be
analyzed analytically, yet are rich enough to teach us something
about the more realistic models. 

A very broad class of such toy models, which we call \emph{$q$-state
models}, can be defined on finite graphs. We consider a graph
$G=(V,E)$ and view its vertices as the microscopic subsystems that
we wish to model; for example, the atoms of a crystal. We assume
that each microscopic system can be found in one of $q$ possible
states that are numbered by $0,1,\ldots, q-1$. We will often refer
to these states as ``colors'', and to the labeling of all vertices
as a ``coloring'' of $G$. Such a coloring is denoted by a vector
$\sigma=(\sigma_1, \sigma_2, \ldots , \sigma_{|V|})$ that assigns a
color $\sigma_i$ to every vertex $i$. A coloring completely
specifies the \emph{microscopic state} of the system. We remark that
we do not require adjacent vertices of $G$ to be labeled by
different colors.

Next, we use the edges of the graph to denote \emph{interactions}
between the vertices. For every edge $e=(i,j)\in E$, we define a
(real) function $h_{ij}(\sigma_i,\sigma_j)$ (also denoted by
$h_e(\sigma_i, \sigma_j)$) that specifies the
interaction energy between the vertices $i$ and $j$.  The overall
energy of the system for a particular coloring is therefore
\begin{align}
  \label{def:energy}
  \mathcal{H}(\sigma) \EqDef \sum_{(i,j)\in E} h_{ij}(\sigma_i, \sigma_j) \ .
\end{align}

To understand the macroscopic behavior of the system, we would like
to know the probability of the system to be in a given microscopic
state. Usually, one assumes that the system is attached to another,
much bigger, system, which we call a ``heat bath''. The attachment
of the two systems means that energy can freely flow from one system
to the other. In such a case, under fairly simple assumptions that
will not be discussed here (see, for example, \Ref{ref:Cal85}), we
find that the probability of the system to be in a microscopic state
$\sigma$ is given by the Boltzmann-Gibbs distribution
\begin{align}
  \Pr(\sigma) = \frac{1}{Z(\beta)} e^{-\beta \mathcal{H}(\sigma)}  \ ,
  \qquad \beta = \frac{1}{k_B T} \ .
\end{align}
Here, $\beta$ is an external parameter, which is inversely
proportional to $T$, the temperature of the system, and $k_B$ is the
Boltzmann constant. $Z(\beta)$ is the normalization factor, which is
called the \emph{partition function} of the system, and is
given by
\begin{align}
\label{def:Z}
  Z(\beta) \EqDef \sum_\sigma e^{-\beta \mathcal{H}(\sigma)} \ .
\end{align}
It turns out that many interesting macroscopic properties of the
system can be deduced solely from the partition function. These
include the average energy of the system, its entropy, specific
heat, and more elaborate properties such as phase-transitions
\cite{ref:Cal85}. The calculation of the partition function is
therefore an important task in the theory of statistical physics.

Before explaining how this can be done in the framework of a tensor
network, we list some of the well-known models of statistical
mechanics that fall into this category, which were also discussed in
Refs.~\cite{ref:Nes06, ref:Nes07}:
\begin{enumerate}
  \item \textbf{Ising Model}\\
    In the Ising model every vertex can be colored one of two
    colors, or alternatively, every vertex denotes a classical
    spin that can point up or down. In order to keep the notation
    simple, let us assume that $\sigma_i$ holds the values
    $\{1,-1\}$ (instead of $\{0,1\}$). Then the interaction energy
    of every edge is simply: 
    \begin{align}
      h(\sigma_i, \sigma_j) = -J\sigma_i\sigma_j \ .
    \end{align}
    $J$ is called the \emph{coupling constant}. If $J>0$, the model
    is called \emph{ferromagnetic}. In this case, the neighboring
    spins will tend to point to the same direction.  When $J<0$, the
    model is called \emph{antiferromagnetic}, and the spins will
    tend to be antialigned. 
    
  \item \textbf{Clock Model}\\
    The $q$-state Clock model is a generalization of the Ising model
    for $q$ colors. Here, at each vertex the spin can point in one
    of $q$ equally spaced directions in the plane and are therefore
    specified by an angle
    \begin{align}
      \theta_n = \frac{2\pi n}{q} \ , \quad n=0, 1, \ldots, q-1 \ .
    \end{align}
    Then the interaction energy is
    \begin{align}
      h(\sigma_i, \sigma_j) = -J\cos \big( \theta_{\sigma_i} -
      \theta_{\sigma_j}) \ .
    \end{align}

  \item \textbf{Potts Model}\\
    The the $q$-state Potts model is another generalization of the
    Ising model, simpler than the clock model. As in the clock
    model, every vertex can be in one of $q$ colors, but instead of
    using the cosine function for the interaction energy, we use the
    Kronecker delta-function:
    \begin{align}
      h(\sigma_i, \sigma_j) = -J\delta_{\sigma_i,\sigma_j} \ .
    \end{align}
\end{enumerate}

The Ising, Clock, and Potts models are all examples of a class of
$q$-state models in which the coupling energies only depend on the
difference (modulo $q$) of the colors, i.e.  $h(\sigma_i, \sigma_j)
= h( (\sigma_i-\sigma_j) \bmod q)$; we shall call the class of
models with this property \emph{difference models}. In the above
three models, the coupling constant $J$ was the same for every edge,
and we say that these systems have \emph{homogeneous coupling}. In
general, different couplings can be used. 

\subsection{Constructing a tensor-network in the general case}

In this section we define a tensor network that
evaluates the partition function of the general $q$-state model.
Applying Theorem \ref{thm:main} to this tensor network will give a
quantum algorithm for approximating the partition function of the
general $q$-state model (Corollary \ref{col:q-state}).  We make no
assumptions on the functional form of the coupling energy functions
$h_e(\sigma_i, \sigma_j)$. In the next section we will assume that
$h(\sigma_i, \sigma_j) = h( \sigma_i - \sigma_j |_{\!\!\!\mod q})$ -
an assumption that is satisfied by the Ising model, the Potts model,
and the Clock model. For difference models, we can write a different
tensor-network which in turn leads to a different quantum algorithm
for approximating the partition function of difference models
(Corollary \ref{col:delta}).  For difference models, the
approximation scales in Corollary \ref{col:q-state} and
\ref{col:delta} are incomparable, i.e., there exist choices of
parameters for which each is better than the other. A tensor-network
seems to be a natural tool to calculate the partition function of a
$q$-state model: in both cases, we have a summation over all
possible labelings/colorings. However, in the partition function,
the summation is over a coloring of the \emph{vertices}, whereas is
in the tensor-network, the summation is over the labeling of the
\emph{edges}. To resolve this mismatch, we introduce a new graph
$\tilde{G}$, by putting a new vertex in the middle of every edge of
$G$ (see \Fig{fig:G-tilde}). We call these vertices \emph{energy
vertices}. They are denoted by $V_\epsilon$, and the vertices of $G$
are denoted by $V_G$. Then $\tilde{V}=V_\epsilon\cup V_G$,
$|\tilde{V}| = |V| + |E|$ and $|\tilde{E}| = 2|E|$.  On this graph,
we define the following network:
\begin{deff}[Tensor network for the general $q$-state statistical model]
\label{def:q-state-network}
  For the $q$-state statistical model that is defined on $G=(V,E)$
  with the interaction energy $h_{i,j}(\cdot,\cdot)$ for every
  $(i,j)\in E$ and an inverse temperature $\beta$, we define the
  following tensor network $T(\tilde{G}, M)$:
  \begin{itemize}
    \item \textbf{Graph:} The graph
      $\tilde{G}=(\tilde{V},\tilde{E})$ that is defined above.
    \item \textbf{Labeling:} Every edge $e\in \tilde{E}$ can be 
      colored in one of $q$ possible colors: $0,1, \ldots, q-1$.
    \item \textbf{Tensors:} We have two types of tensors: 
      \begin{enumerate}
        \item For $v\in V_G$ (original vertex of $G$), the tensor 
          $\B{M} _v$ is defined to be zero unless all its edges $G_v$ are
          colored identically, in which case it is $1$.
        \item For $v\in V_\epsilon$ (energy vertex), which is in the 
          middle of an original edge $(i,j)\in E$ of $G$, the tensor
          $\B{M} _v$ is defined by
        \begin{align}
          \label{eq:edge-tensor}
         ( \B{M}_v)_{\sigma_i, \sigma_j} \EqDef e^{-\beta h_{ij}(\sigma_i,
          \sigma_j)} \ .
        \end{align}
        Here $\sigma_i, \sigma_j$ are the coloring of the two edges
        that connect to $v$.
      \end{enumerate}
      
  \end{itemize}
  
\end{deff}

\begin{figure}
  \center \includegraphics[scale=1]{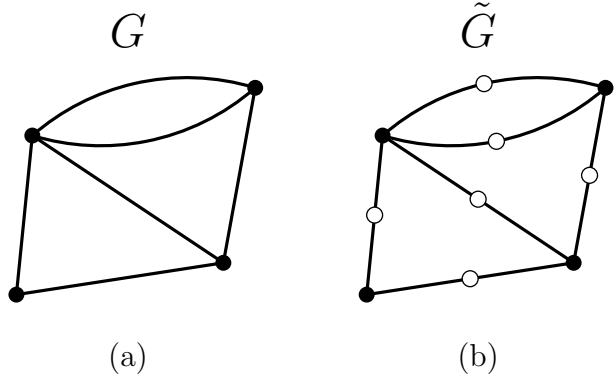} 
  
  \caption{Creating $\tilde{G}$ from $G$ by placing new "energy" vertices
    (unfilled) in the middle of every edge of $G$.
    \label{fig:G-tilde}}
\end{figure}

The following theorem shows that this tensor-network evaluates $Z_G$.
\begin{theorem}
  The tensor-network in \Def{def:q-state-network} evaluates
  $Z_G(\beta)$.
\end{theorem}
\begin{proof}
  By \Eq{eq:alternative}, the value of the tensor network is
  \begin{align}
    T(\tilde{G}, M) = \sum_l \left( \prod_{v\in V_G} \B{M} _v (l)\right) \cdot 
     \left(\prod_{v\in V_\epsilon} \B{M} _v (l)\right) \ .
  \end{align}
  For every labeling $l$, the term $\left( \prod_{v\in V_G}
  \B{M} _v (l)\right)$ vanishes unless the edges connected to each original vertex of $G$
  are labeled identically. This defines a unique coloring of each
  \emph{original vertex of $G$}, which we denote by $\sigma$, and
  therefore
  \begin{align}
    T(\tilde{G}, M) = \sum_{\sigma} \left(\prod_{v\in V_\epsilon}
    \B{M} _v (\sigma)\right) \ .
  \end{align}
  Then by \Eq{eq:edge-tensor}, we obtain
  \begin{align}
    T(\tilde{G}, M) = \sum_\sigma \prod_{e\in E} 
      e^{-\beta h_{ij}(\sigma_i, \sigma_j)}
    = \sum_\sigma e^{-\beta\mathcal{H(\sigma)}}=Z_G(\beta) \ .
  \end{align}
  
\end{proof}

We would now like to analyze the
approximation scale of the above tensor-network that results from a
simple bubbling. To do that, let us first agree on a notation to
describe the swallowing of a vertex: we say a vertex with $n+m$
adjoint edges is swallowed in a $n\to m$ fashion if before it is
swallowed it has $n$ adjoint edges that end inside the bubble, and
$m$ edges that end outside the bubble. Figure~\ref{fig:swallowing}
illustrates this notation.
\begin{figure}
  \center \includegraphics[scale=0.7]{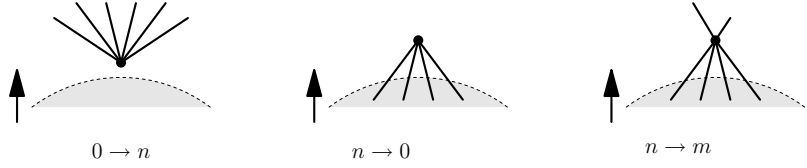}
  \caption{The $n\to m$ notation for describing the swallowing of a vertex.
  \label{fig:swallowing}}
\end{figure}
Consider now the following simple bubbling: we embed the parent
graph $G$ in a 3D space such that every vertex is put at a different
height and every edge is a straight line. Thus all edges are
non-horizontal. We then add the energy vertices in the middle of
every edge of $G$. Our bubbling is defined by swallowing $\tilde{G}$
using an horizontal plane (bubble) that moves from bottom to top.
The embedding ensures that every energy vertex is swallowed in a
$1\to 1$ fashion. The original vertices of $G$ can be swallowed
in many different ways. We set  $b$ to be the number of such vertices that
are swallowed in a $0\to n$ or a $n\to 0$ fashion.

Then the norms of this bubbling are as follow:
\begin{itemize}
  \item \textbf{Original vertices ($V_G$)}
    \begin{itemize}
      \item for the $b$ vertices that are swallowed in a  $0\to n$
        or $n\to 0$ fashion, the norm is $q^{1/2}$. Indeed in a
        $0\to n$ swallowing, the operator the normalized state
        $\ket{\Omega}$ and creates the state $\ket{0}^{\otimes
        n}+\ket{1}^{\otimes n}+\ldots+\ket{q}^{\otimes n}$, whose
        norm is $q^{1/2}$. Similarly, a $n\to 0$ accepts any state
        $\ket{\psi}$ and outputs the state $c\ket{\Omega}$, where
        $c$ is the inner product of $\ket{\psi}$ with
        $\ket{0}^{\otimes n}+\ket{1}^{\otimes
        n}+\ldots+\ket{q}^{\otimes n}$. Obviously, when $\ket{\psi}$
        is normalized, $|c|\le q^{1/2}$.  
        
      \item For the vertices that swallowed in a $n\to m$ fashion,
        with $n>0, m>0$, the norm is $1$. To see this, expand a
        normalized state in the domain $H^{\otimes n}$ of the
        operator in terms of the standard basis. The only terms in
        the expansion that will not be annihilated are those of the
        form $\ket{i}^{\otimes n}, i=0, \ldots, q-1$; they will be
        sent to $\ket{i}^{\otimes m}$ in the range of the operator.
        The operator is therefore the identity operator from a
        $q$-dimensional subspace in the domain to a $q$-dimensional
        subspace in the range, hence its norm is 1.  
    \end{itemize}

    \item \textbf{Energy vertices ($V_\epsilon$)}. These are always
    swallowed in a $1\to 1$ fashion. In that case, the tensors act 
     as a $q\times q$ matrix $\left( e^{-\beta h_e(\sigma_i,
    \sigma_j)} \right) _{\sigma_i,\sigma_j}$ 
    that maps the color space of one edge to the color
    space of the other edge; the norm of the tensor is the
    operator norm of the matrix. 
\end{itemize}    

Combining all of this together, we have an approximation scale of:
\begin{equation}  
    \Delta \le q^{b/2} \prod_{e\in E} \norm{e^{-\beta h_e}} \ .
\end{equation}

We now give an upper bound on $b$. Clearly $b\leq |V|$, but in many
embadings of $G$ in $\BBR^3$, $b$ is significantly smaller than
$|V|$.  Consider the 2 dimensional lattice; rotate the lattice so
that two opposite corners are the highest and lowest vertex of the
graph. Now the bubbling described above only swallows two
vertices in $V_G$ (the first and the last) in a $0 \rightarrow n$ or
$n \rightarrow 0$ fashion and thus $b=2$.  The identical argument
holds for any finite lattice of any dimension. More generally, we
can describe $b$ in terms of directed acyclic graphs on $G$.
\begin{definition}  
  Given a graph $G=(V, E)$, a \emph{directed acyclic graph on
  $G$} is an assignment of direction to each edge $E$ so that
  there are no directed cycles.  An \emph{extreme vertex} of a
  directed acyclic graph on $G$ is a vertex for which the adjacent
  edges are either all directed towards the vertex, or all directed
  away from the vertex. 
\end{definition}
It should be clear that any bubbling induces a directed acyclic
graph on $G$, directing any edge away from the vertex that is
swallowed first.  With this point of view, $b$ is then the number of
extreme points of the directed acyclic graph.
 
Putting this all together we have the following:

\begin{corollary}[An efficient quantum algorithm for the general $q$-state
  model] 
\label{col:q-state}
  Consider a $q$-state model defined over a finite graph $G=(V,E)$,
  and a directed acyclic graph on $G$ with $b$ extreme vertices. Then
  there exists an efficient quantum algorithm for an additive
  approximation of the partition function of the general $q$-state
  model that is defined over the finite graph $G=(V,E)$ with
  coupling energies $h_e(\sigma_i, \sigma_j)$. The approximation
  scale of the algorithm is 
  \begin{align} 
    \label{eq:q-state-err}
    \Delta \le q^{b/2} \prod_{e\in E} \norm{e^{-\beta h_e}} \ .
  \end{align}
  Here, $\norm{e^{-\beta h_e}}$ denotes the operator norm of the
  $q\times q$ matrix $\left( e^{-\beta h_e(\sigma_i, \sigma_j)}
  \right) _{\sigma_i,\sigma_j}$. 
  
  When $G$ is a finite lattice (of any dimension), there is always
  an acyclic graph with $b=2$, hence the resultant approximation
  scale is:
  \begin{align} 
    \Delta \le q\prod_{e\in E} \norm{e^{-\beta h_e}} \ .
  \end{align}   
\end{corollary}

\begin{figure}
  \center 
  \includegraphics[scale=1]{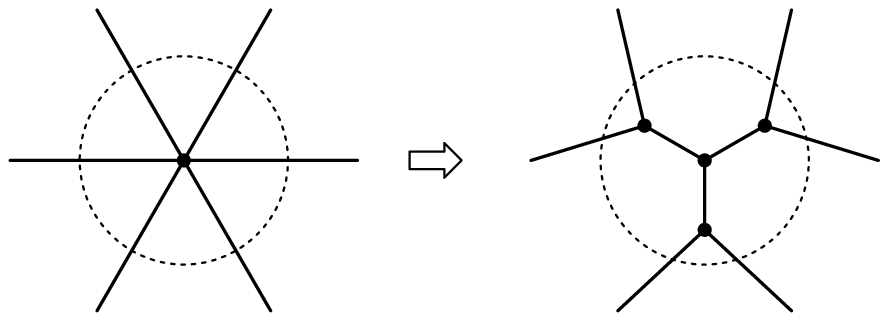}
  \caption{Reducing the degree of a vertex by replacing it with a
  low-degree tree. This reduction is possible for simple vertices
  such as the identity vertex and the cycle vertices in
  \Sec{sec:difference}.
  \label{fig:deg-reduction}}
\end{figure}

We  note that:
\begin{enumerate}

  \item\label{pg:deg-reduction} In the corollary we did not restrict the 
    underlying graph $G$ to have a bounded degree. The reason is
    that the identity tensors are \emph{reducible} in the following
    sense: every vertex of degree $k$ that represents an identity
    tensor can be locally replaced by a tree graph with bounded
    degree (say, $d=3$) as described in \Fig{fig:deg-reduction}. The
    new vertices of the tree are defined to be identity tensors as
    well, and so by the connectivity of the tree, the only
    non-vanishing coloring of the external edges is the one where
    they are all colored identically. Moreover, in such case, their
    overall weight is 1 as required. 
    
    This reduction does not affect the approximation scale,
    for if we swallow the identity tensor in an $m\to n$ manner with
    $m>0, n>0$ then we can also swallow the tree graph such no
    vertex is swallowed in a $0\to \ell$ or $\ell\to 0$ manner.
    Therefore the norm of all the identity tensors is $1$, hence
    their product is also $1$ - in agreement with the original
    scale. If, on the other hand, we swallow it in a $0\to k$ or
    $k\to 0$ manner, then we can swallow the tree graph such that
    one vertex is swallowed in a $0\to \ell$ manner while the rest
    are not, thereby yielding the overall approximation scale
    $q^{1/2}$ as required.

  \item The above algorithm works for any
        temperature $\beta$ and any coupling energies $h_e(\cdot,
        \cdot)$ - not necessarily physical ones (e.g., they can be
        complex).
        
 \end{enumerate}

Shortly after the first release of this paper, Van den Nest et al
\cite{ref:Nes08} published an interesting paper with closely related
results. They gave quantum algorithms for several classical
statistical-mechanics models and also provided some complementary
hardness results. The models considered were on two-dimensional
grids with restrictions on the form of the coupling energies.
Specifically, they considered two types of statistical models:
vertex models and edge models \cite{ref:Bax08}.  Edge models are
essentially $q$-state models, in which the particles sit on the
vertices of a graph and the edges model the interactions between
them. These include the Ising model, which was the actual model of
this type that the authors studied.  In the vertex models, the
classical particles sit on the edges of the grid, and the vertices
model many-particles interactions.

These algorithms fit nicely in the tensor-network formalism that we
presented here. Indeed, a close inspection reveals that they can be
considered as a bubbling of a two-dimensional network on a grid, and
therefore, as a special case of \Thm{thm:main}. Moreover, in the
Ising model case, the functional form of the tensors they use is
identical to the one used here. In addition, in all cases the
bubbling is along one dimension of the grid (say, from left to
right), and the couplings in the model are restricted so that 
swallowing operators are unitary (thereby directly implementable on
a quantum computer). Finally, the authors assumed that configuration
of the particles on the left and right boundaries of the grid is
fixed. Consequently, the approximation scale in both models is of
order one.

To prove the \BQP-hardness of these approximations,
the authors showed that in both models, one can choose the couplings
such that the resulting unitary operators form a universal set of
gates for quantum computation. We can use their hardness result
with respect to the Ising model to prove the following Corollary: 
\begin{corollary}[\Ref{ref:Nes08}]
  There exist classes of general $q$-state models for which the
  approximation achieved in Corollary $\ref{col:q-state}$ is
  BQP-hard.
\end{corollary}
\begin{proof}
  As mentioned above, in the Ising case, \Ref{ref:Nes08} uses the
  same tensor-network as in \Def{def:q-state-network} and places it
  on a two-dimensional grid that is swallowed form left to right.
  The only difference is that they allow the use of fixed boundary
  conditions. This means that the identity tensors on the right
  and left boundaries are replaced by tensors that allow only one
  configuration out the possible $q$. Consequently, their network is
  incompatible with our general $q$-state model construction.
    
  Nevertheless, we can turn their fixed boundary Ising model into a
  general $q$-state model of the type considered here by
  introducing two additional vertices, one to the left and one to
  the right of the lattice.  Connecting each of these to the
  boundary, we can impose the fixed boundary conditions by an
  appropriate choice of couplings of the new edges. For example,
  assume that the additional vertex is indexed by $i$ and that $j$
  sits on the boundary. If we want to fix the color of boundary
  $j$ to 0, we define the interaction between these particles
  by $e^{-\beta h(\sigma_i, \sigma_j)}$ that is $1$ for
  $\sigma_i=\sigma_j=0$ and $0$ otherwise. This tensor-network is
  therefore no-longer a ``pure'' Ising model, but it is still a
  general $q$-state model. In addition, its value is identical to
  the value of the tensor network in \Ref{ref:Nes08}.
  
  Swallowing this network from left to right, we notice that $b=2$,
  and that the norm of the newly introduced edges is exactly 1.
  Since the rest of the energy vertices in the network produce
  unitary operators with a unit norm, Corollary $\ref{col:q-state}$
  achieves a constant approximation scale $\Delta$. By
  \Ref{ref:Nes08}, this network evaluates the result of a universal
  quantum computation, and therefore the approximation is \BQP-hard.
\end{proof}

\subsection{Tensor networks for the Difference Models.} \label{sec:difference}

In this section we restrict our attention to difference models.  As we noted earlier, the Ising, Clock, and Potts
models are all examples of difference models.  For these models, we
can define an alternatative tensor network, that together with a
suitable bubbling, evaluates the partition
function with a better approximation scale than the general
case for certain choices of couplings. The idea is to use the redundancy of the coupling energy to
define the tensor network on a smaller Hilbert space, which 
leads to smaller norms, hence, a better approximation scale.

We begin by turning $G$ into a directed graph
by placing an arrow on every edge of $G$. The direction of the arrow
is unimportant. Then for
every edge $e=(i,j)$ with an arrow going from
$j$ to $i$, we define a variable $\delta_{i,j}\EqDef (\sigma_i -
\sigma_j) \bmod q $. $\delta_{i,j}$
takes on the values $0,1, \ldots, q-1$. By assumption, the coupling
energy $h_{i,j}(\sigma_i, \sigma_j)$ depends only on that variable.

We would like to write the partition function as a sum over
all possible labeling of the delta variables. However, these
variables are not independent: every cycle in the graph (a cyclic
sequence of adjacent vertices) yields a \emph{consistency
constraint} on the variables that are associated with its edges;
their appropriate sum or difference must yield 0. For example, in
\Fig{fig:cycles}, there are 3 cycles: 
\begin{itemize}
  \item $1\to 2\to 3\to 1$: corresponds to $\delta_{2,1} +
    \delta_{3,2} + \delta_{1,3} = 0$.
  \item $1\to 4\to 3\to 1$: corresponds to $\delta_{4,1} -
    \delta_{4,3} + \delta_{1,3} = 0$.
  \item $1\to 2\to 3\to 4\to 1$: corresponds to $\delta_{2,1} +
    \delta_{3,2} + \delta_{4,3} - \delta_{4,1} = 0$.
\end{itemize}
\begin{figure}
  \center \includegraphics[scale=1]{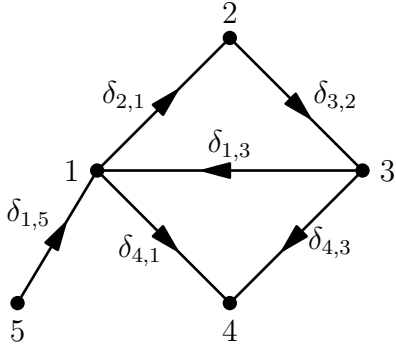} \caption{The cycles
  in a graph and their relation with the constraints over the
  $\{\delta_{i,j}\}$ variables. In this example, there are 3 cycles
  $1\to 2\to 3$, $1\to 4\to 3\to 1$, and $1\to 2\to 3\to 4\to 1$ -
  but only 2 of them are independent. They are translated into 3
  consistency equations between the $\{\delta_{i,j}\}$ variables
  (with only two being independent) as described in the text.
  \label{fig:cycles}}
\end{figure}
Not all these equations are independent; the third equation is the
difference of the first two, just as the first two cycles can be
joined to obtain the third. We should therefore limit our attention
to a set of \emph{independent cycles} that correspond to a set of
independent equations. One possible way to obtain such a set is the
following: start with a spanning tree of the graph. As $G$ is
connected, the spanning tree must have $|V|-1$ edges. Now add the
remaining $|E|-|V|+1$ edges. Every such edge $e=(i,j)$ creates a new
cycle that only uses $e$ and part of the original spanning tree
(since the $i$ and $j$ vertices were already connected by the
spanning tree). Moreover, this cycle is independent of the other
cycles because none of them contain the edge $(i,j)$. This way we
construct a set of $|E|-|V|+1$ independent cycles, which we denote
by $\mathcal{C} = \{ C_i\}$. The number of independent $\delta$
variables is therefore $|V|-1$, one for each edge of the original
spanning tree. Finally, it is easy to verify that this procedure can
be done efficiently. 

The following lemma shows that we can sum over all labeling of
$\{\delta_{i,j}\}$ that satisfy $\mathcal{C}$ in order to obtain the
partition function.
\begin{lemma}
\label{lem:delta-coloring} 

  Let $G=(V,E)$ be a directed graph on which a $q$-state difference
  model is defined with the coupling functions
  $h_{i,j}(\delta_{i,j})$. Then its partition function is given by
  the following sum over labeling of the delta variables, which
  satisfy the consistency constraints:
  \begin{align}
    Z_G(\beta) = q\sum_{\substack{\text{consist'}\\ \text{labeling}}}
      \prod_{(i,j)\in E}e^{-\beta h_{i,j}(\delta_{i,j})} \ .
  \end{align}
\end{lemma}
\begin{proof}
  We show that every consistent labeling of the $\{\delta_{i,j}\}$
  variables comes from exactly $q$ different colorings of the
  vertices of $G$.  Given a consistent labeling of the
  $\{\delta_{i,j}\}$ variables, pick a vertex $v_0\in V$ and assign
  to it a color $\sigma_{v_0}\in 0,1,\ldots, q-1$. Then define the
  color of every other vertex $v\in V$ by following a path from
  $v_0$ to $v$ and subtracting/adding the right $\delta_{i,j}$
  variables along the path. The coloring of $v$ is independent of
  the actual path, for otherwise two paths that produce different
  coloring of $v$ would create a cycle whose appropriate summation
  of the $\delta$ variables would not vanish. We have therefore used
  the labeling of $\{\delta_{i,j}\}$ to define a coloring $\sigma$
  of the vertices of $G$.  As the coloring we have just defined
  depends on the initial assignment of $v_0$, there are at least $q$
  labelings that produce $\{\delta_{i,j}\}$.
  
  These are also the only colorings that do that. Indeed, consider
  some coloring of the vertices. This coloring has some assignment
  to the vertex $v_0$ that corresponds to one of the $q$ colorings 
  that we have defined. By following the path from $v_0$ to any
  other vertex $v\in V$, it is easy to see that the two colorings
  must agree on all vertices and not only on $v_0$.
\end{proof}

With this result we are in a position to define a tensor network
that is based on the $\delta$ variables. The first step is to define
the graph $G_\delta$ on which the network is defined

\begin{deff}[The graph $G_\delta$]
  
  The graph $G_\delta=(V_\delta, E_\delta)$ is constructed from a
  graph $G=(V,E)$ as follows.  Pick a spanning tree for $G$.  Denote
  the set of edges of the spanning
  tree by $E_{\text{tr}}$, and let $E_{\text{cycle}}\EqDef E
  \backslash E_{\text{tr}}$.  As previously discussed, the choice
  of a
  spanning tree produces a set of independent cycles $\mathcal{C} =
  \{ C_e \}_{e \in E_{\text{cycle}}}$ ($C_e$ involves $e$ and some
  subset of the edges of $E_{\text{tr}}$). Then we construct
  $G_\delta$ in 4 steps, which are also illustrated in
  \Fig{fig:delta-tnet}.
  \begin{description}
    \item[(a)] We embed $G$ in three dimensions and identify its cycles. We
      will use its edges and vertices only as a guide for the
      construction of $G_\delta$.
      
    \item [(b)] We place a vertex in the middle of every edge
      of $G$; vertices that correspond to $e\in E_{\text{tr}}$ are
      called \emph{tree vertices}, and those that correspond to
      $e\in E_{\text{cycle}}$ are called \emph{cycle vertices}.
      Above every such vertex we create another vertex, which is
      called an \emph{energy vertex}, and connect the two vertices with
      an edge.
      
    \item [(c)] For every cycle $C_e$, we connect the cycle vertex 
      corresponding to $e$ with all tree vertices associated to
      edges in $E_{\text{tr}}$ involved in the cycle $C_e$.
          
    \item [(d)] We put a vertex in the middle of every edge that
      connects a tree vertex to its energy vertex. These are called
      \emph{mid vertices}. For every cycle $C_e$ we connect the
      cycle vertex associated to $e$ with all mid vertices
      associated to edges in $E_{\text{tr}}$ involved in the cycle
      $C_e$.
  \end{description}
\end{deff}

We now describe a tensor network on $G_{\delta}$ that will
evaluate the partition function. The network has the following
tensors:
\begin{itemize}  
  \item \emph{Tree vertices and mid vertices}.  These are identity 
    tensors: they are zero unless all the colors of their edges are
    equal, in which case they are 1.
    
\item \emph{Energy vertices}.  These vertices have only one edge,
    which is associated with a $\delta$ variable. Their definition
    is $\B{M}_v(\delta)=e^{-\beta h_e(\delta)}$, with
    $h_e(\cdot)$ being the corresponding energy function of the
    edge in the parent graph $G$.
    
\item \emph{Cycle vertices}.  Recall that the edges of a cycle 
  vertex come in pairs that correspond to the tree vertices in that
  cycle: one edge connects the cycle vertex to the tree vertex and
  the other connects the cycle vertex to the associated mid vertex.
  
  The tensor shall be zero unless the labels of each pair are equal.
  When they are equal, we interpret the label of each pair to be the
  $\delta$ value of the underlying tree edge (of the original
  graph). In addition, we interpret the label of the edge that
  connects the energy vertex as the $\delta$ variable of underlying
  cycle edge. When all these labelings satisfy the consistency
  equation of the cycle, the tensor is $1$. Otherwise it is zero.
  
\end{itemize}

\begin{figure}
  \center \includegraphics[scale=0.7]{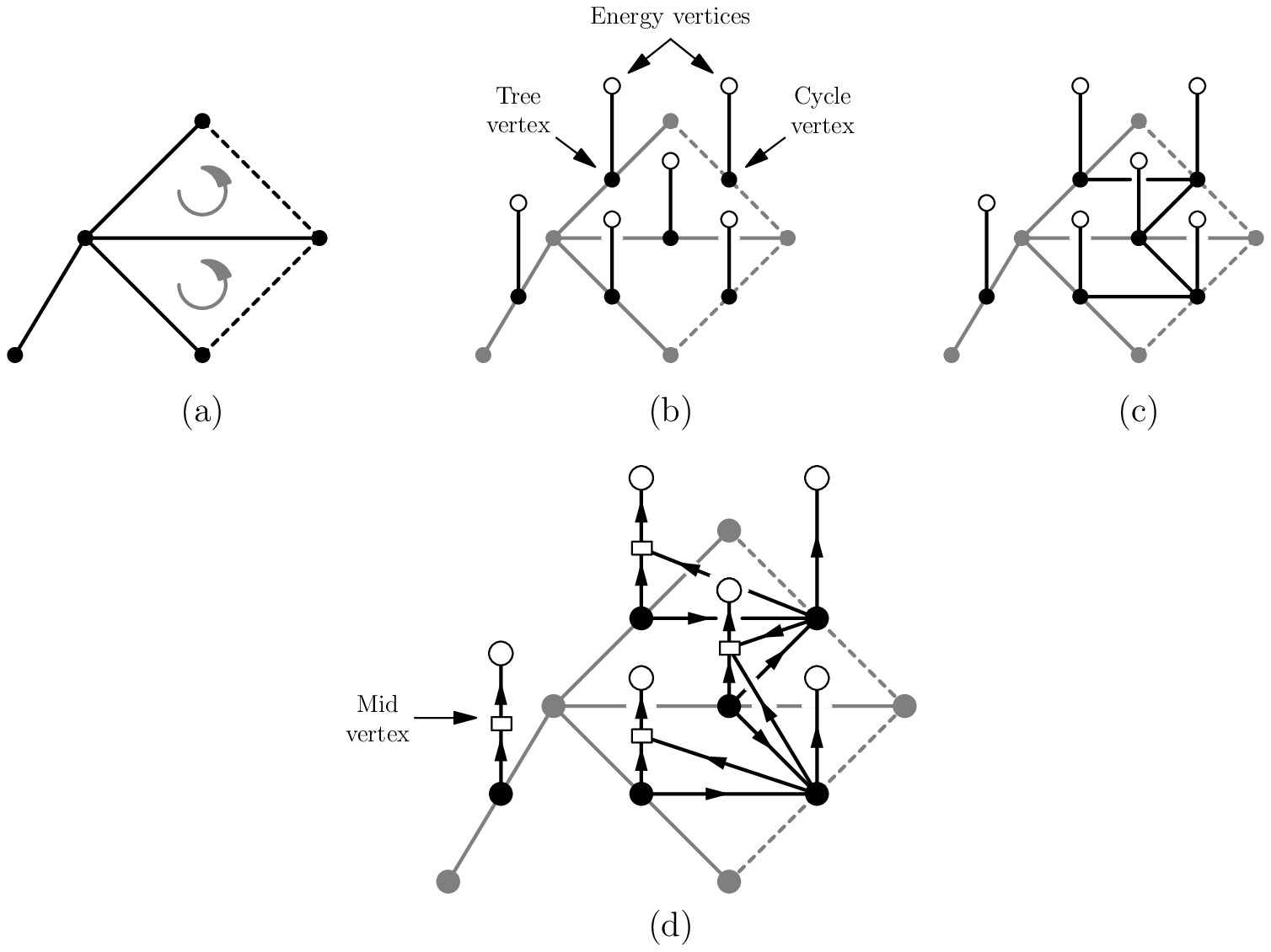}
  \caption{Constructing $G_\delta=(V_\delta, E_\delta)$ from
  $G=(E,V)$. (a) We begin with a connected graph $G=(E,V)$ with a
  spanning tree that is denoted by the solid edges. $G$ contains
  two independent cycles, denoted by arrows. (b) We place a tree
  vertex in the middle of every tree edge of $G$ and a cycle vertex in the
  middle of every cycle edge (black filled vertices). We connect
  them to energy vertices (unfilled vertices), which are placed
  above. (c) We connect the cycle vertices to the tree vertices in
  their cycle. (d) We place a mid vertex (unfilled square) in the
  middle of every edge that connects a tree vertex to its energy
  vertex. Finally we connect the cycle vertices to the mid vertices
  in their cycle. The arrows on edges denote the bubbling order of
  $G_\delta$. \label{fig:delta-tnet}}
\end{figure}

The following lemma shows that the above tensor network evaluates
the partition function.
\begin{lemma}
  The tensor-network that was defined above for the graph
  $G_\delta=(V_\delta, E_\delta)$ evaluates $q^{-1}Z_G(\beta)$, with
  $G=(V,E)$ being the original graph from which $G_\delta$ was
  constructed.
\end{lemma}
\begin{proof}
  According to Lemma~\ref{lem:delta-coloring}, it is enough to show
  that the tensor network gives
  \begin{align}
    \label{eq:del-sum}
    \sum_{\substack{\text{consist'}\\ \text{labeling}}}
      \prod_{(i,j)\in E}e^{-\beta h_{i,j}(\delta_{i,j})} \ ,
  \end{align}
  where the sum is over all labelings of the delta variables
  $\{\delta_{ij}\}$ that satisfy the consistency constraints. 
  
  The tensors of the tree vertices are identity tensors. Therefore
  in a non-vanishing labeling of $G_\delta$, all edges that connect
  to a tree vertex must have the same labeling. This uniquely
  defines a labeling of the tree vertices. The converse is also
  true: any labeling of the tree vertices uniquely defines a
  non-vanishing labeling of $G_\delta$. Indeed, given a labeling of
  the tree vertices, we first label all their
  incident edges. This way every mid vertex has exactly one labeled
  edge, which determines the labeling of the rest of the edges. The
  only edges which are left unlabeled are the edges that connect the
  cycle vertices to their energy vertices. Their labeling is
  uniquely determined by the cycle constraint as manifested by the
  tensors of the cycle vertices.
  
  Now every consistent labeling of the $\{\delta_{i,j}\}$ variables
  in $G$ is uniquely determined by a labeling $\{\delta_{i,j}\}$ of
  the tree edges, which is equivalent to a labeling of the tree
  vertices. Therefore every consistent labeling of
  $\{\delta_{i,j}\}$ corresponds to a non-vanishing labeling of
  $G_\delta$ a vice-versa.  We leave it to the reader to verify that
  that for these non-zero terms, the value of the network is exactly
  \begin{align*}
    \prod_{(i,j)\in E}e^{-\beta h_{i,j}(\delta_{i,j})} \ .
  \end{align*}
\end{proof}

Let us now analyze the approximation scale of this tensor network in
a simple bubbling that is described in Fig.~\ref{fig:delta-tnet}
(d).  We place the 4 types of vertices on 4 different horizontal
planes: all the tree vertices are put on the lowest plane. In the
plane above we place the cycle vertices, followed by the mid
vertices and finally the energy vertices.  By the definition of
$G_{\delta} $, all edges are inter planar. Therefore by bubbling the
graph from bottom to top, we have 4 types of norms:
\begin{itemize}
  \item Tree vertices: the bubbling here is $0\to n$.
    The tensor is an identity tensor and therefore the norm is $q^{1/2}$.
    
  \item Cycle vertices: the bubbling here is $n \rightarrow n+1$,
    where the first $n$ connect to the tree vertices and the second
    $n+1$ connect to the corresponding mid vertices as well as to
    the energy vertex. The labeling of the input edges uniquely
    determines the labeling of the output edges with a weight of
    unity. Therefore the norm is $1$.
  
  \item Mid vertices:  the bubbling here is $n\rightarrow 1$ for these 
    identity tensors which yields a norm 1.
  
  \item Energy vertices: bubbling here is from $1\to 0$. It is easy
    to see that the norm here is $\left(\sum_{j=0}^{q-1}|e^{-\beta
    h(j)}|^2\right)^{1/2}$.
\end{itemize}
Multiplying these norms together and using the fact that there are
exactly $|V|-1$ tree vertices, we arrive at the following
corollary:

\begin{corollary}
  \label{col:delta0}
  The above tensor-network and bubbling yields an additive
  approximation to a $q$-state difference model with the
  approximation scale
  \begin{align} 
  \label{eq:delta-scale}
    \Delta = q^{(|V|+1)/2} \prod_{e\in E}\left(
      \sum_{j=0}^{q-1}|e^{-\beta h_e(j)}|^2\right)^{1/2} \ .
  \end{align}      
\end{corollary} 
Notice that in the formula above we have a factor $q^{(|V|+1)/2}$
instead of $q^{(|V|-1)/2}$ because our tensor network evaluates
$q^{-1}Z_G$ and therefore we must multiply its approximation scale
by a factor of $q$.

In recent work, Van den Nest et al~\cite{ref:Nes06, ref:Nes07}
revealed an interesting link between the partition function of
classical statistical models and quantum physics. They show how to
express the $q$-state partition functions of the difference models
as inner products between certain graph states and a simple tensor
product state. Since both states can be efficiently generated by a
quantum circuit, their paper contained an implicit quantum algorithm
to approximate these partition functions \cite{ref:Comm07}; analysis
of the additive error produces the identical scale to
\Eq{eq:delta-scale}.  Even more recently, during the time that the
work presented here was being refereed, Van den Nest used this
connection in an interesting work \cite{ref:Nes09} to show that this
scale is classically achievable.  However, the approximation scale
in \Eq{eq:delta-scale} can be improved to a scale for which no
classical simulation result is known.  The improvement that we
describe now serves as an example of the flexibility that the
tensor-network formalism offers.

We start by combining the energy vertices of the cycle vertices into
the cycle vertices. Then the new tensor is non-vanishing if and only
if the original tensor is non-vanishing, only that now the
non-vanishing configurations are given the appropriate energy
weight. The bubbling of this new vertex is in $n\to n$ fashion, and
its norm is $\max _{j} |e^{-\beta h_e (j)}|$. This is strictly
smaller than the combined contribution of the original two vertices,
which is $\left(\sum_{j=0}^{q-1}|e^{-\beta h_e(j)}|^2\right)^{1/2}$.

The second modification is to redistribute the weight of the energy
vertices that correspond to the tree vertices. We split it equally
between the energy vertex and its associated tree vertex: when these new 
vertices have all their edges labeled by the same $j$, their weight 
is $\sqrt{e^{-\beta H_e(j)}}$. Under the
same bubbling as before, the contribution of the two vertices
becomes $\sum_{j=0}^{q-1} |e^{-\beta h_e(j)}|$, which is smaller
than or equal to the previous contribution
$q^{1/2}\left(\sum_{j=0}^{q-1}|e^{-\beta h_e(j)}|^2\right)^{1/2}$.
All in all, the new approximation scale yields the following result:

\begin{corollary}[An efficient quantum algorithm for the difference $q$-state
  model] 
  \label{col:delta} 
  
  Given a difference $q$-state model on a graph $G=(E,V)$, and
  spanning tree $E_{tr} \subset E$, with $E_{cycle} \EqDef
  E\backslash E_{tr}$, there exists an efficient quantum algorithm
  that provides an additive approximation of the partition function
  $Z_G(\beta)$ with the approximation scale

\begin{align}
  \label{eq:delta-scale2}
    \Delta = q \left( \prod_{e\in E_{\text{tr}}}
      \sum_{j=0}^{q-1}|e^{-\beta h_e(j)}|\right)\cdot
      \left(\prod_{e \in E_{\text{cycle}}} 
         \max_{j} |e^{-\beta h_e (j)}| \right) \ .
\end{align}
\end{corollary}

Notice that unlike the previous approximation scale, this scale
depends on the spanning tree. It is always smaller than or
equal to the scale in \Eq{eq:delta-scale}.  We do not know of a
complementary hardness result, hence we cannot generally assess the
quality of the approximation.  However, the classical simulation
results of \cite{ref:Nes09} that apply to the approximation scale
describe in \Eq{eq:delta-scale}, cannot readily be applied to
achieve the approximation scale in Corollary \ref{col:delta}
\cite{ref:Comm09}.

Just as in the general case, we are not limited to physical
energies, and the functions $h_e(\delta)$ and $\beta$ can be
complex.  Similarly, we have not restricted the shape of the
original graph $G$ because the resulting high-degree vertices in
$G_\delta$ are always reducible: they are either associated with the
identity tensors, which, as explained in the first remark in page
\pageref{pg:deg-reduction}, are reducible, or they are cycle
vertices, which are also reducible by a similar argument.

It is natural to wonder whether the approximation scale given by
Corollary \ref{col:delta} yields better approximations than
specializing the more general Corollary \ref{col:q-state} to
difference models.  This will depend on the specific model and
parameters being considered.  For all choice of parameters we have
the following inequality: 
\begin{align*}
  \max_j |e^{-\beta h_e (j)}| \leq
    \Big\|\left( e^{-\beta h_e((i-j) \! \! \mod q)}
      \right)_{i,j}\Big\|
    \leq \sum_{j=0}^{q-1}|e^{-\beta h_e(j)}| \ .
\end{align*}
Since there exist parameters to achieve both extremes in this
inequality, and there are graphs for which $b=2$ in Corollary
\ref{col:q-state}, it follows that there exist certain models for
which the approximation scale in Corollary \ref{col:delta} is better
and others for which Corollary \ref{col:q-state} is better.

We conclude this section by applying Corollary \ref{col:delta} to
the special case of the $q$-state Potts model with an homogeneous
coupling $J$. In that case, $e^{-\beta h_e(j)}=e^{\beta J}$ for
$j=0$, and $e^{-\beta h_e(j)}=1$ for $0<j<q$. Therefore the spanning
tree dependence disappears, and we obtain
\begin{corollary}
  \label{col:Potts} 
  There exists an efficient quantum algorithm that gives an additive
  approximation of the partition function $Z_G(\beta)$ of the
  homogeneous $q$-state Potts model that is defined on an arbitrary
  graph $G=(V,E)$ with inverse temperature $\beta>0$ and a coupling
  constant $J$. The approximation scale is given by
  \begin{align}
    \Delta = \left\{
      \begin{array}{lcl}
        q\left(q-1+e^{\beta J}\right)^{|V|-1} 
          \left(e^{\beta J}\right)^{|E|-|V|+1}&,&
          \text{Ferromagnetic case ($J>0$)} \\ \\
        q\left(q-1+e^{\beta J}\right)^{|V|-1} &,&
          \text{Antiferromagnetic case ($J<0$)}
      \end{array} \right. \ .
  \end{align}
\end{corollary}

It is interesting to compare the above results to a classical result
that is given in Proposition 5.2 of \Ref{ref:Bor05}. There, the
authors assert that there exists straightforward classical sampling
algorithm that provides an additive approximation for the Tutte
polynomial $T_G(x,y)$ of a connected graph $G$ for $x>1, y>1$ with
an approximation scale $y^{|E|}(x-1)^{|V|-1}$. However, for such
graphs, $T_G(x,y)=(x-1)(y-1)^{|V|}Z_G(\beta)$, where $Z_G(\beta)$ is
the partition function of the homogeneous Potts model with
$q=(x-1)(y-1)$ and $y=e^{\beta J}$ \cite{ref:Jae90, ref:Sok05}. Therefore their
classical algorithm provides an additive approximation for the
ferromagnetic case with an approximation scale
$\Delta'=q^{|V|}\big(e^{\beta J}\big)^{|E|}$, and we obtain
\begin{align}
  \frac{\Delta}{\Delta'} 
    = \left(\frac{q+e^{\beta J}-1}{q e^{\beta J}}\right)^{|V|} \ .
\end{align}
This ratio is exponentially small as long as $q>1$ and $\beta J >
0$, and thus it may be seen as an indication
for the non-triviality of our approximation.
On the other hand, this classical result is in many cases better
than the quantum results of Corollary~\ref{col:q-state} and
Corollary~\ref{col:delta} for the homogeneous ferromagnetic Potts
case, which questions their non-triviality in the other cases, and
emphasizes the crucial role of the bubbling.

Finally, \Ref{ref:Bor05} also presents a simple additive
approximation for the chromatic polynomial $P_G(q)$ that counts the
number of legal $q$-colorings of the graph $G$. They show that there
exists a classical additive approximation for $P_G(q)$ for integer
$q$'s, whose approximation scale is $(q-1)^{|V|}$. Moreover, this
scale is tight, in the sense that for any $0<\delta<q-1$, an
additive approximation with scale $(q-1-\delta)^{|V|}$ is \NP-hard.

It is easy to verify that the chromatic polynomial is obtained from
the antiferromagnetic partition function of the homogeneous Potts
model for $e^{\beta J}\to 0$ (i.e., $J<0$ and $\beta\to +\infty$).
Not surprisingly, in such case the approximation scale of \Col{col:Potts} is
equivalent to $(q-1)^{|V|}$ -- the classical result.

\section{Tensor networks and the Jones and Tutte polynomials.}
\label{sec:top}

Recently, efficient quantum algorithms have been given for
additively approximating certain topological/combinatorial
quantities: the Jones polynomial of braids at roots of unity
\cite{ref:Fre02a, ref:Fre02b, ref:Fre02c, ref:Aha06b} and the Tutte
polynomial of planar graphs \cite{ref:Aha07}. Broadly speaking, both
results can be viewed in three steps:
\begin{enumerate}
  \item The problem is mapped into a combinatorial calculation within 
    the Temperley-Lieb algebra.
  \item Representation theory of the Temperley-Lieb algebra is used 
    to translate the combinatorial problem into a linear-algebra problem.
  \item A quantum algorithm is given for approximating the  solution 
    to the linear-algebra problem.
\end{enumerate}

This final step can be seen as the approximation of a particular
tensor network.  Without going into the details, the rough
description of the tensor network for the two problems is as
follows:
\begin{itemize}
  \item \textbf{The Jones Polynomial of a Braid.}  Here the tensor 
    network is derived from the braid by closing up the loose
    strands of the braid and then replacing every crossing by a
    vertex corresponding to a rank 4 tensor, and inserting a vertex
    at any local maximum or minimum of the strands.

  \item \textbf{The Tutte Polynomial of a planar Graph.}  Here, the 
    original graph $G$ is replaced by a so called \emph{medial graph},
    which features a rank-four tensor at the center of every edge of
    $G$.  Again, as in the Jones Polynomial case, rank two tensors
    are inserted at any local maximum or minimum.
\end{itemize}

Even with this rough description, an intuitive understanding of the
nature of the errors given in these works can be obtained.  In the
Jones Polynomial case, the parameter being a root of unity ensures
that the rank 4 tensors can be swallowed $2 \rightarrow 2$ such that
the swallowing is a unitary operator and hence does not effect the
scale $\Delta$.  What remains is the cost of swallowing the maximum
and minimum tensors in a $0 \rightarrow 2$ or $2 \rightarrow 0$
fashion, each of which contributes a factor of $\sqrt{q}$.  In the
Tutte Polynomial case, the contributions to $\Delta$ include the
previous cost of the rank 2 tensor swallowing, but in addition,
unlike the Jones Polynomial case, also include a cost for each $2
\rightarrow 2$ swallowing of the rank 4 tensors (in the language of
this paper, these quantities are the $||\rho (\mathcal{T}_i) ||$
terms). This is because in \Ref{ref:Aha07}, the crossing operators
are not necessary unitary.

We have previously discussed the need to carefully examine the
nature of the additive error that \Thm{thm:main} provides. In the
context of these problems, the non-trivial nature of the
approximation has been established by showing that for certain sets
of parameters, the level of approximation provided by both
algorithms has been shown to be a complete problem for quantum
computation \cite{ref:Fre02a, ref:Yar06, ref:Aha06a, ref:Aha07}.

\section{Conclusions and open questions}
\label{sec:summary}

We have given a quantum algorithm that additively approximates the
value of a tensor network to a certain scale.  As an application of
the algorithm, we have obtained new quantum algorithms that
approximate the partition functions of certain statistical
mechanical models including the Potts model.

The fact that the approximation is additive and depends on the
approximation scale is by no means a minor point: for a given
algorithm, with large enough approximation error, the approximation
is useless and the algorithms are trivial, or at least can be
matched classically.  We have shown that in some cases, the
approximation scale of the algorithm is \BQP-hard, and therefore,
some instances of the problem are highly nontrivial.  We consider
this to be an important but indirect verification that the
approximation scale is non-trivial.  What is missing is a direct
argument: an argument that would say the approximation scale is good
enough to answer a question directly connected to the quantities
being estimated (i.e. topological invariants, statistical mechanical
models).  Such an argument would represent a significant advance.

Our intuition is that the tensor network point of view should be
helpful for the design of new quantum algorithms in the future. This
is motivated by the fact that from the tensor network viewpoint two
core features of quantum circuits, the unitarity of the gates, and
the notion of time (i.e. that the gates have to be applied in a
particular sequence), are replaced by more flexible features.  The
unitary gate is replaced by an arbitrary linear map encoded in each
tensor and the notion of time is replaced by the geometry of the
underlying graph of the tensor network along with a choice of
bubbling of the network.  Hence, the design of algorithms from the
tensor-network point of view requires two things: a tensor network
whose value is the quantity of interest, and a specification of a
bubbling order of the vertices of the underlying graph.  It often
seems that for a specific problem there are several somewhat natural
tensor networks with the right value to choose from and that the
approximation scale can vary quite dramatically between the choices
(as was the case in the difference statistical mechanical models).
Additionally, for a given tensor network, the choice of ordering can
make a significant difference as well.  The analysis of these issues
have a combinatorial and graph theoretic flavor.  It would be
interesting to understand the computational complexity of finding an
optimal or even a reasonably good choice of bubbling for a given
tensor network.

It is also intersting to see if the tensor-network framework can
help to understanding other models of quantum computation. For
example, in the one clean qubit model of quantum computation, we
have a universal quantum computer that is allowed to operate only on
one clean qubit (initialized to $\ket{0}$), while the rest of qubits
are in a completely mixed state $\rho=\Id$ \cite{ref:Kni98}. The
result of such a computation can be seen as the trace of a product of 4
operators: a quantum circuit $U$ times local projection $Q$, times
$U^\dagger$ and another local projection $P$. In the tensor-network
setting, the trace of these 4 operators will look like edges
connecting one part of the chain to its other part. It is therefore,
not surprising that the estimation of the Jones polynomial of the
trace closure of a braid is known to be in this model (in fact,
recently Shor \& Jordan have shown that it is complete for this
model \cite{ref:Sho08}), because the trace closure of a braid, when
interpreted as a tensor network, translates to the (weighted) trace
operation in a quantum computation. However, as the same
tensor-network can be graphically presented and bubbled in many
different ways, it might be hard to identify the trace operation
that hides in a particular layout of the network. It
would therefore be interesting to see if there exists a more natural
way to characterize these networks, using, perhaps, some property of
the network that is invariant to the way in which it is presented.
Such characterization might shed some more light on this interesting
complexity class.

It is also interesting to consider the idea of additive
approximations from a complexity theory point of view. In a recent
paper, Goldberg \& Jerrum studied the complexity of a multiplicative
approximation (FPRAS) for the Tutte polynomial \cite{ref:Les07}.
They map about three quarters of the Tutte plane, distinguishing
between points where there is an FPRAS and points where an FPRAS is
NP-hard. It would be interesting to do the same with respect to
additive approximations. In light of the recent quantum algorithms
for the Jones and Tutte polynomials, as well the results of this
paper, it seems that additive approximations are a natural framework
for quantum algorithms. We therefore hope that quantum hardness
results can be used to map regions in the Tutte plane which are
inaccessible to classical additive approximations with certain
approximation scales (unless $\BPP=\BQP$). It is also interesting to
understand the relationship between such points and other points and
approximation scales where an additive approximation is $\NP$-hard.
The first few steps in that direction were taken by \Ref{ref:Bor05},
and we hope that the algorithms and techniques of this paper can be
used to further advance these ideas.

Finally, we briefly mention two other directions of inquiry that
might be of interest.  The first is to see whether there is a
natural extension of our tensor-network definition of the $\BQP$
class to the $\QMA$ class (or more likely the $\QCMA$ class). Can
such a definition shed new light on these complexity classes? A
related problem is to find a $\QMA$-complete problem that is
\emph{naturally} cast in the language of tensor-networks.  

The second direction is to understand the structure of universal
sets of tensors, i.e. sets of elementary tensors that can be
efficiently contracted to approximate any other
tensor. So far, such sets were found solely using techniques from
quantum computations: one begins with a set of transformations that
form a dense subgroup $SU(N)$ or $SL(N)$, and then proves
universality using the (either unitary or non-unitary - see
\Ref{ref:Aha07}) Solovay-Kitaev theorem. The set of universal
transformations yields a set of universal tensors.  It is therefore
interesting to see if there exist other, perhaps more direct,
techniques to prove such universality, techniques that do not rely
on heavy machinery from the theory of Lie groups.

\section{Acknowledgments}
\label{sec:Acknowledgements}

We would like to thank Maarten van den Nest, Wolfgang D\"{u}r and
Hans Briegel for pointing our attention to an implicit quantum
algorithm for the evaluation of the partition function of the Potts
model using the framework of graph states and one-way quantum
computation. Their observation paved the way for the difference algorithms in
\Sec{sec:difference}. We would also like to thank Dorit
Aharonov for useful discussions and suggestions.

\bibliographystyle{hep}
\bibliography{tnet}

\appendix

\section{The Hadamard test}
\label{sec:Hadamard}

The Hadamard test is a simple and well-known quantum algorithm that
approximates the inner product $\bra{\alpha}U\ket{\alpha}$ for a
normalized state $\ket{\alpha}$ and a unitary operator $U$ that can
be efficiently generated. As it is an important part of our main
result \Thm{thm:main}, we include it here for sake of
completeness. 

\begin{theorem}[The Hadamard test]
\label{thm:Hadamard}
  Let $\ket{\alpha}$ be a normalized state that can be efficiently
  generated (e.g., a tensor product $\ket{0}^{\otimes n}$), and let
  $U$ be a unitary operator that can be implemented on a quantum
  computer in time $T$. Then there exists a quantum algorithm 
  that for every $\epsilon>0$ outputs a complex number $r$ such that
  \begin{align} \label{e:hadtest}
    \Pr\Big( |\bra{\alpha}U\ket{\alpha} - r| \ge \epsilon\Big) \le
    1/4 \ ,
  \end{align}
  and the running time of the algorithm is $\CO(\epsilon^{-2} T)$.
\end{theorem}

\begin{proof}
  We add an ancillary  qubit to the system and initialize it in the state
  $\ket{\psi_0}=\ket{\alpha}\ket{0}$. Acting with the Hadamard gate
  $H=\frac{1}{\sqrt{2}}\bigl( \begin{smallmatrix} 1&-1\\ 1&1
  \end{smallmatrix} \bigr)$  on the ancillary qubit, we get
  $\ket{\psi_1} = \frac{1}{\sqrt{2}}\ket{\alpha}\otimes(\ket{0} +
  \ket{1})$. The next step is to act with $U$ on the registers of
  $\alpha$ \emph{conditioned} on the anciallary qubit. The result is
  $\ket{\psi_2} = \frac{1}{\sqrt{2}}\big(\ket{\alpha}\otimes\ket{0} +
  (U\ket{\alpha})\otimes\ket{1} \big)$. Finally we act again with the
  Hadamard gate on the ancillary qubit and obtain
  \begin{align}
    \ket{\psi_3} = \frac{1}{2}\Big[ 
      \ket{\alpha}\otimes\ket{0} + \ket{\alpha}\otimes\ket{1} 
      + (U\ket{\alpha})\otimes\ket{1} - (U\ket{\alpha})\otimes\ket{0}
      \Big] \ .
  \end{align}
  We measure the ancillary qubit and output the number $1$ for $\ket{1}$ and
  $-1$ for $\ket{0}$. We repeat this process $N$ times and store the
  results in the variables $x_1, \ldots, x_N$. These are independent
  identically distributed random variables with an average
  $E(x_i)=\text{Re}\bra{\alpha}U\ket{\alpha}$ since
 $\Pr(x_i=1) =
  \frac{1}{4}[2+2\text{Re}\bra{\alpha}U\ket{\alpha}]$ and $\Pr(x_i=-1) =
  \frac{1}{4}[2-2\text{Re}\bra{\alpha}U\ket{\alpha}]$. We can
  therefore use the Chernoff-Hoeffding bound and obtain
  \begin{align}
    \Pr\Big[ \big|\frac{1}{N}\sum_{i=1}^N x_i -
    \text{Re}\bra{\alpha}U\ket{\alpha}\big| \ge \epsilon\Big] \le
    2e^{-2N\epsilon^2} \ .
  \end{align}
  Thus taking $N=\CO(\epsilon^{-2})$, we obtain the right approximation
  for $\text{Re}\bra{\alpha}U\ket{\alpha}$. 

  To approximate the imaginary part, we change the first step such
  that $\ket{\psi_1} = \frac{1}{\sqrt{2}}\ket{\alpha}\otimes(\ket{0} -
  i\ket{1})$, and proceed in the same way. All in all, the entire
  algorithm runs in $\CO(T\epsilon^{-2})$ quantum time.
  
  Notice that we can replace the $1/4$ factor in (\ref{e:hadtest}) by any constant
  $\delta>0$ and obtain a running time of
  $\CO(T\epsilon^{-2}\log\delta)$.
\end{proof}

\end{document}